\documentclass[11]{article}

\pdfoutput=1
	\usepackage[colorlinks,citecolor=blue,urlcolor=blue,linkcolor = blue]{hyperref}
	\newcommand{\red}[1]{{\textcolor{black}{#1}}}	
	\usepackage[comma,sort&compress]{natbib}
	\usepackage{amssymb, amsmath,bm,graphicx}
	\usepackage{setspace}
	\usepackage{mathrsfs,amsthm}
	\usepackage[normalem]{ulem}
	\usepackage{subcaption}
	\usepackage{comment}
	\usepackage{times,bm}
	\usepackage{multirow}
	\usepackage{anyfontsize}
	\usepackage{graphicx,psfrag,epsf}
	\usepackage{float}
	\usepackage[font=small,labelfont=bf]{caption}
	\usepackage{tabularx}
	\usepackage{url}
	\usepackage{booktabs}
	\usepackage{hyperref}
	\usepackage{color}
	\usepackage{xcolor}
	\usepackage{caption}
	\usepackage{bm}
	\usepackage{bbm}
	\usepackage[mathscr]{euscript}
	\usepackage{bigstrut,enumerate}
	\usepackage[utf8]{inputenc}
	\usepackage[ruled,vlined]{algorithm2e}
	\usepackage{multirow,pkgfile}
	\renewcommand{\liminf}{\varliminf}
	\renewcommand{\limsup}{\varlimsup}
	
	\newtheorem{theorem}{Theorem}
	
	\newtheorem{lemma}{Lemma}
	\newtheorem{proposition}{Proposition}	
	
	\newtheorem{remark}{Remark}
	\newtheorem{corollary}{Corollary}

	\newcommand{\mbb}[1]{\mathbb{#1}}
	
	\newcommand{\mc}[1]{\mathcal{#1}}

	\DeclareMathOperator*{\argmin}{arg\,min}

\usepackage{color,soul}	
		
	\newcommand{\blind}{1}
	
\theoremstyle{definition}
\newtheorem{defn}{Definition}

\newtheorem{assumption}{Assumption}

\SetKwInOut{Parameter}{Parameter}

\begin{document}
\if1\blind
{
\title{\vspace{-2.5cm}A Nearest-Neighbor Based Nonparametric Test for Viral Remodeling in Heterogeneous Single-Cell Proteomic Data}
\author{$\text{Trambak Banerjee}^{1}, \text{Bhaswar B. Bhattacharya}^{2} \text{ and Gourab Mukherjee}^3$\\
$^{1}$University of Kansas, $^{2}$University of Pennsylvania and $^{3}$University of Southern California}
\date{}
\footnotetext[3]{The research here was partially supported by NSF DMS-1811866.}
\footnotetext[3]{Corresponding author: gmukherj@marshall.usc.edu}
\maketitle
\vspace{-1.0cm}
} \fi

\if0\blind
{
\title{\vspace{-2.5cm}{Nonparametric Two Sample Testing Under Heterogeneity With Applications to Single-Cell Virology}}
\date{}
\maketitle
\vspace{-1.0cm}
} \fi
\begin{abstract}
An important problem in contemporary immunology studies based on single-cell protein expression data is to determine whether cellular expressions are remodeled post infection by a pathogen. 
One natural approach for detecting such changes is to use nonparametric two-sample statistical tests.  However, in single-cell studies, direct application of these tests is often inadequate,  because single-cell level expression data from processed uninfected population\red{s} often contain attributes of several latent subpopulations with highly heterogeneous characteristics. As a result, viruses often infect these different subpopulations at different rates in which case the traditional nonparametric two-sample tests for checking similarity in distributions are no longer conservative. In this paper, we propose a new nonparametric method for \textit{Testing Remodeling Under Heterogeneity} (TRUH) that can accurately detect changes in the infected samples compared to possibly heterogeneous uninfected samples.  
Our testing framework is based on composite nulls and is designed to allow the null model to encompass the possibility that the infected samples, though unaltered by the virus, might be dominantly arising from under-represented subpopulations in the baseline data. The {TRUH} statistic, which uses nearest neighbor projections of the infected samples into the baseline uninfected population, is calibrated using a novel bootstrap algorithm. 
We demonstrate the non-asymptotic performance of the test via simulation experiments, and also derive the large sample limit of the test statistic, which provides theoretical support towards consistent asymptotic calibration of the test. We use the TRUH statistic for studying remodeling in tonsillar T cells under different types of HIV infection and find that unlike traditional tests which do not have any heterogeneity correction, TRUH based statistical inference conforms to the biologically validated immunological theories on HIV infection. 
\end{abstract}

\noindent \textbf{Keywords:\/} single-cell virology, immunology, two-sample tests, viral remodeling, homogeneous Poisson process, nearest neighbors, HIV infection, mass cytometry. 

\newpage
\section{Introduction}
In many contemporary scientific methodologies, it is extremely difficult, even in well-regulated laboratory experiments, to simultaneously control the multitude of factors that give rise to heterogeneity in the population (Chapter 3 of \citet{holmes2018modern}).
Nevertheless, these experiments are very powerful, and are often our only recourse to study several interesting biological phenomena. 
For example, in single-cell  proteomic and 
genomic studies (\citet{jiang2018single,wang2018gene,jia2017accounting,shi2017identifying}),  it is now well understood that there is high heterogeneity in cellular responses from controlled cell population. Statistical tests are often used on these datasets to determine differences between the case and control samples. The presence of heterogeneity greatly complicates statistical inference, and direct application of existing two-sample testing methods, without modulating for the latent heterogeneity in the samples, may lead to erroneous statistical decisions and scientific consequences. The problem of testing similarity in the distributions of two samples under heterogeneity arises in a host of modern immunology research set-ups where heterogeneous  protein expression datasets collected at single-cell resolution are analyzed to detect viral perturbation. 
\red{To provide a rigorous statistical hypothesis testing framework for these immunology studies, we consider a composite null hypothesis that allows mixture expression distributions in cases and controls with the mixture having same components but potentially different mixing proportions; the alternative hypothesis contains scenarios where at least one of the mixture components is actually different between the cases and the controls.}
We develop a new nonparametric testing procedure based on nearest-neighbor distances, that can accurately detect if there are differences between the case and control samples in the presence of unknown heterogeneity in the data-generation process.  We next provide the background of the problem through an immunology study on human immunodeficiency virus (HIV) infection in tonsillar cells. 

\subsection{Phenotypic Profiling of T Cells Under HIV Infection}  
\label{sec:motivation}
In single-cell immunology, phenotypic profiling of immune cells under the influence of a target virus, such as the HIV \citep{cavrois2017mass}, the varicella zoster virus (VZV) \citep{sen2014single}, or the rotavirus (RV) \citep{sen2012innate}, is a critical research endeavor. It enhances understanding of which subsets of cells are most or least  susceptible to infection, leading to new insights regarding the magnitude of viral persistence, which is crucial in the development of life saving drugs \citep{sen2015single}.  Mass cytometry based techniques \citep{Bendall11,giesen2014highly} are popularly used for generating proteomic datasets for such phenotypic analysis. These techniques can simultaneously measure around fifty protein expressions on individual cells. In this paper we provide a rigorous statistical analysis for testing if there are any HIV induced changes in the proteomic expressions of tonsillar T cells, which are a type of lymphocyte that plays a central role in the immune response, based on the dataset generated in \citep{cavrois2017mass}.

Figure \ref{fig:cytofcartoon} presents a schematic representation of the experimental set-up used  for generating single-cell level proteomic expression data of HIV infected T cells using Cytometry by Time Of Flight (CyTOF) technique.
Tonsillar T cells from 4 healthy donors were infected with two variants of a HIV viral strain: \texttt{Nef rich HIV} and \texttt{Nef deficient HIV}. \texttt{Nef} (Negative Regulatory Factor) is a protein encoded by HIV which enhances virus replication in the host cell by protecting infected cells from immune surveillance. 
We study the differential impact of these two variants on the immune cells. 
The healthy cells were cultured \red{and processed into three batches} 
for each donor. For each patient, one among the three batches were randomly selected and phenotyped to generate the expression data of the uninfected population, while the other 
two batches  were contaminated with the \texttt{Nef rich HIV} and the \texttt{Nef deficient HIV}, respectively, and phenotyped after $4$ days of infection.
All the batches where phenotyped using multi-parameter CyTOF panel which contained $35$ surface markers and $3$ viral markers. These are special proteins attached to the cell membrane. After leaving out dead cells from each run of the CyTOF experiment we had $38$ protein expressions for approximately $25,000$ uninfected cells. Virus infected cell in the contaminated population were marked based on the expression of the viral markers and it was found that the number of virally infected cells in the batch subjected to HIV infection was around $250$. These cells constitute the infected cell population. 
%
\begin{figure}[!h]
	\centering
	\includegraphics[width=0.95\linewidth]{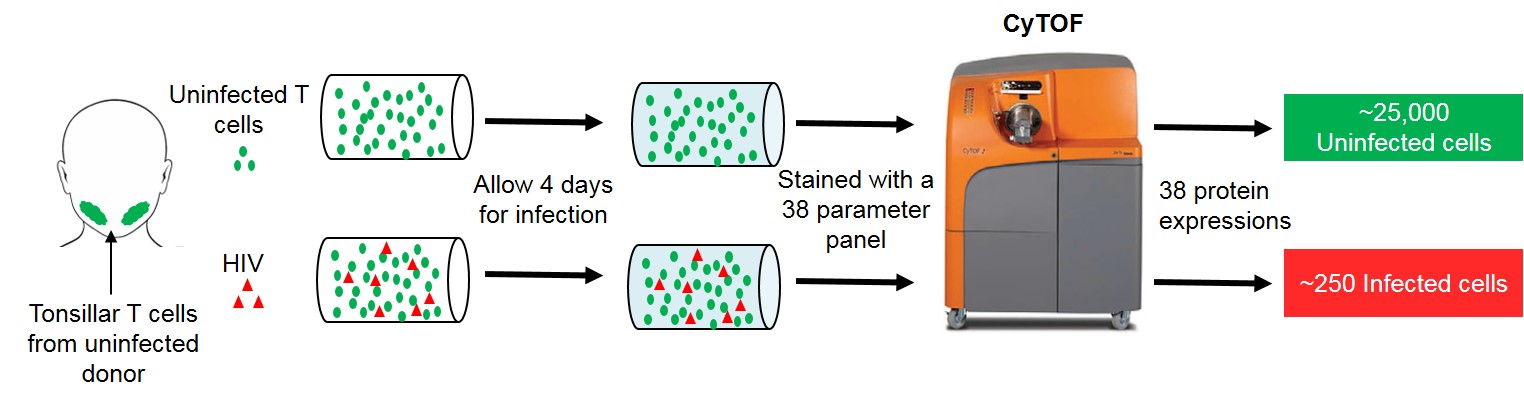}
	\caption{\small{Schematic representation of the experimental design associated with the phenotypic analysis of HIV infected CD4+ T cells using mass cytometry. Tonsillar T cells from a healthy donor (represented by green circles) are infected with the \texttt{Nef rich} or the \texttt{Nef deficient} HIV virus (represented by red triangles). These cells were then phenotyped in a $38$ parameter panel after allowing 4 days for infection. The resulting data has $38$ protein expressions for approximately $25,000$ uninfected cells and the number of virally infected cells was around $250$.} \\[-6ex]}
	\label{fig:cytofcartoon}
\end{figure}
\subsection{Viral Remodeling} 

If the virus changes the expression of any of the surface markers, which are proteins attached to the cell membrane, then the cell is said to have undergone viral remodeling of its phenotypic characteristics \citep{sen2014single}.  A virally remodelled cell will have aberrant inter-cellular activities, therefore, detecting the presence of remodeling is a fundamental step towards understanding the mechanism of pathogenesis and disease progression. Detecting remodeling translates to testing if there is enough evidence in the data to reject the null hypothesis that the joint distribution of all the surface proteins is same between the uninfected and virus infected sample. A natural approach for this problem is to invoke nonparametric two-sample testing methods to see if there is enough evidence to support the alternative hypothesis that the virus has changed the distribution of least one of the subpopulations. However, for single-cell level expression data, the hypothesis test described above is particularly difficult because of the following two reasons: (a) the presence of \textit{heterogeneity} in the uninfected population, and (b) due to the phenomenon of \textit{preferential infection}. Single-cell resolution expression data from processed uninfected population often contains attributes from several latent subpopulations with highly heterogeneous  characteristics. This subpopulation level heterogeneity in the uninfected (also referred to as the control or baseline) samples can arise from varied attributes that cannot be controlled in experiments, such as differences in  the cell effector functions, trafficking and longevity \citep{cavrois2017mass}. Viruses often infect these different subpopulations at different rates. If a virus infects different subpopulation at different rates, but does not alter the marker expressions for any of the subpopulations, still the distribution of the overall viral sample will be different from the uninfected samples. In these situations, the difference in distribution between the infected and the uninfected samples is not due to \textit{viral remodeling} but due to \textit{preferential infection} (for a detailed biological explanation {see Figures 2A and 2B of \citep{cavrois2017mass}}) of the uninfected subpopulations by the virus. 
\par
Figure \ref{fig:remodelingcartoon} presents two scenarios that may arise when the cloud of infected and uninfected cells are analyzed with respect to a single marker $A$. In this toy example, Panel 1 in Figure \ref{fig:remodelingcartoon} shows that the uninfected T cells arise from three subpopulations with varying expression levels for marker $A$ which may reflect their inherent heterogeneity with respect to cell longevity. 
\begin{figure}[!t]
	\centering
	\includegraphics[width=0.85\linewidth]{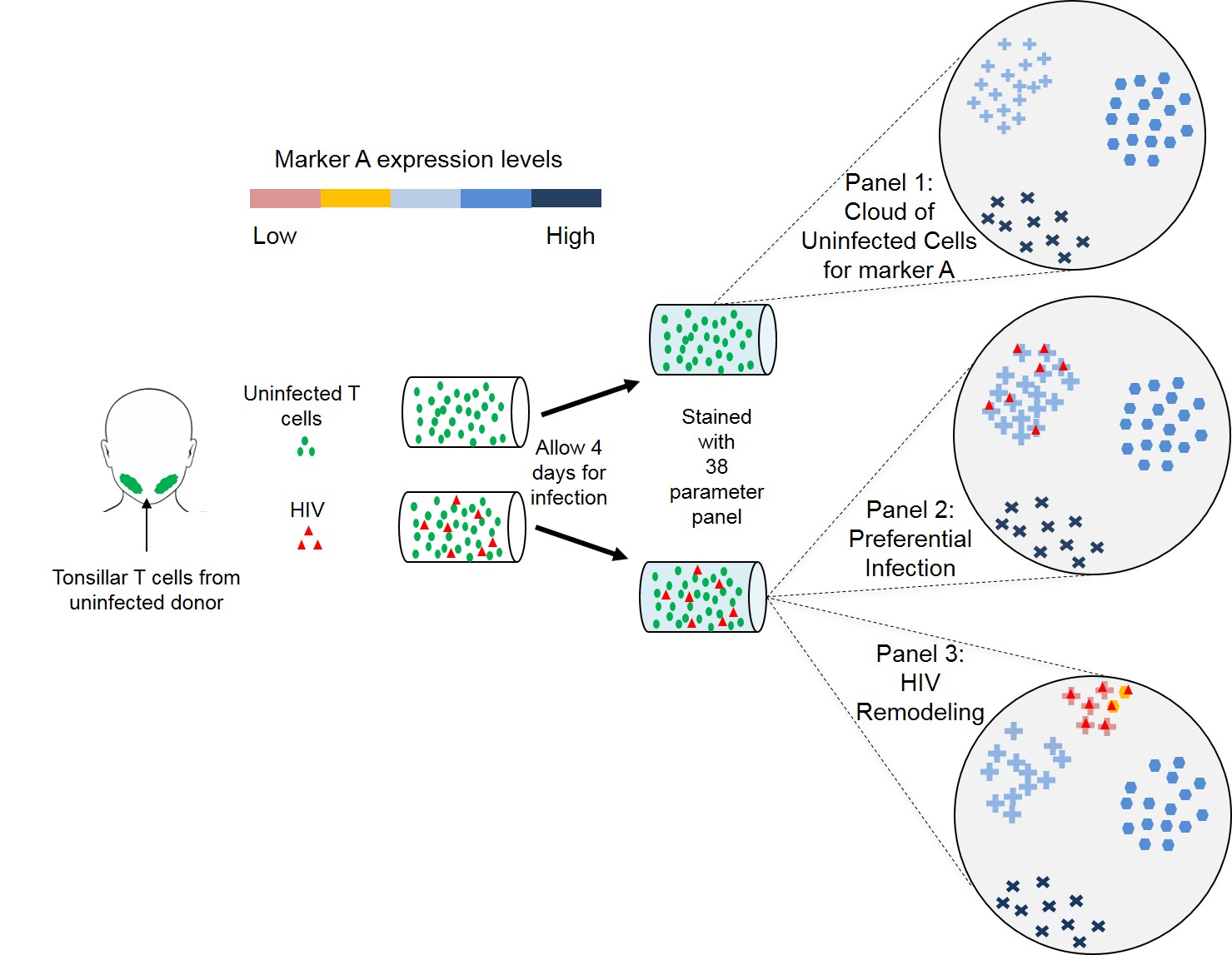}
	\caption{\small{Schematic representation of HIV remodeling of T cells with respect to a single marker $A$. Panel 1 shows that the uninfected T cells arise from three subpopulations with varying expression levels for marker $A$. Panel 2 depicts \textit{preferential infection} where the HIV preferentially infects the T cell subpopulation that has a lower expression level for marker $A$ amongst the uninfected cells and the infection does not alter the expression levels of the T cells when compared to Panel 1. Panel 3 represents \textit{HIV remodeling} where the HIV targets those uninfected cells that have low to medium expression for marker $A$ amongst the uninfected cells and alters their original expression levels upon infection, \red{which is represented by the distinct pink and yellowish shade of the infected cells.}}\\[-6ex]}
	\label{fig:remodelingcartoon}
\end{figure} 
The scenario of \textit{preferential infection} is depicted in Panel 2 where the HIV preferentially infects the T cell subpopulation that has a lower expression level for marker $A$ amongst the uninfected cells. Moreover, the virus does not alter the expression levels of these infected cells when compared to Panel 1. In Panel 3, which represents \textit{HIV remodeling}, the virus targets those uninfected cells that have low to medium expression for marker $A$ amongst the uninfected cells and alters their original expression levels upon infection. The distinct pink and yellowish shade of the infected cells in panel 3 depicts their phenotypic change associated with infection. Here, we have described the phenomenon of viral remodeling only for the HIV. However, remodeling analysis is widely conducted across virology for understanding mechanism of other pathogens also. For correct scientific understanding of the viral mechanism, it is extremely important to accurately \red{distinguish} 
the instances of viral remodeling from mere preferential infection. 
However, popular single-cell based segmentation and classification algorithms \citep{amir2013visne,bruggner2014automated,Linderman12,Qiu12} lack a rigorous statistical hypothesis testing framework for conducting two-sample inference and can greatly suffer in testing problems, particularly if there is high imbalance in the sizes of the uninfected (control) and infected (case) samples, which is often the situation in virology.

\subsection{Testing Procedures in Existing Literature and Statistical Challenges}
The statistical framework for testing remodeling falls under the realm of nonparametric two-sample testing. For univariate data, nonparametric two-sample tests like the Kolmogorov-Smirnov test, the Wilcoxon rank-sum test, and the Wald-Wolfowitz runs test are extremely popular and find a place in every practitioner's toolkit. Multidimensional versions of these widely used tests date back to the randomization tests of \citet{chung1958randomization} and to the generalized Kolmogorov-Smirnov test of \citet{bickel1969distribution}. Friedman and Rafsky \citep{friedman1979multivariate} proposed the first computationally efficient nonparametric two-sample test, which applies to high-dimensional data. The Friedman-Rafsky edgecount test, which can be viewed as a generalization of the univariate runs test,  computes the Euclidean minimal spanning tree (MST)\footnote{Given a finite set $S \subset \mbb{R}^d$, the {\it minimum spanning tree} (MST) of $S$ is a connected graph  with vertex-set $S$ and no cycles, which has the minimum weight, where the weight of a graph is the sum of the distances of its edges.}  of the pooled sample, and rejects the null if the number of edges with endpoints in different samples is small. Many variants of the edgecount test, based on nearest-neighbor distances and geometric graphs have been proposed over the years \cite{hall2002permutation,henze1984number,rosenbaum2005exact,schilling1986multivariate,weiss1960two}. Recently,  \citet{chen2017new} suggested novel modifications of the edge-count test for high-dimensional and object data, and  \citet{chen2018weighted} proposed new and powerful tests to deal with the issue of sample size imbalance. Asymptotic properties of two-sample tests based on geometric graphs can be studied in the general framework described in \citet{bhattacharya2019}. Other popular two-sample tests include the test of \citet{baringhaus2004new}, the energy distance test of \citet{aslan2005new}, and  the kernel based test using {maximum mean discrepancy} of  \citet{gretton2007kernel}. More recently, \citet{chen2013ensemble} address the problem of sample size imbalances in the two-sample problem by constructing an ensemble subsampling scheme for the nearest-neighbor tests \cite{henze1984number,schilling1986multivariate}. Very recently, \citet{deb_twosample_2019} and  \citet{ghosal2019multivariate} proposed distribution-free two-sample tests based on the concept of multivariate ranks, defined using optimal transport. 
\red{Methods based on nearest neighbor distances have been also used extensively in other non-parametric statistical problems, such as density estimation \cite{mack1983rate,mack1979multivariate}, nonparametric clustering \cite{heckel2015robust}, classification \cite{cannings2017local,cover1967nearest,gadat2016classification,samworth2012optimal}, entropy and other functional estimation \cite{berrett2019efficient_II,berrett2019efficient,kozachenko1987sample} and testing problems, such as testing for normality \cite{vasicek1976test}, testing for uniformity \cite{cressie1976logarithms}, and independence  testing \cite{berrett2019nonparametric,goria2005new}.}

One of the main challenges for devising a statistically correct test to detect viral remodeling from preferential infection is that the virus may infect different subpopulations at different rates. In Section \ref{sec:test}, we show that even in very large sample sizes direct application of existing nonparametric two-sample tests can lead to erroneous inference. We expound this phenomenon by exhibiting explicit scenarios of preferential infection and remodeling where traditional tests fail in a simple setting of $d=2$ markers. In Figure \ref{fig:motfig}, the green triangles correspond to a sample of uninfected (UI) cells that arise from three different subpopulations while the red dots reflect the infected (VI) cells. The leftmost panel presents a setting where the virus has infected all the three cellular subpopulations and the overlap of the UI and VI cells indicate no remodeling. The middle panel presents a scenario where the cells have undergone remodeling under the influence of the virus as is evident through a shift in the location of the VI cells. The rightmost panel reflect no remodeling but preferential infection.
The different \texttt{g-tests} \citep{chen2018weighted,chen2017new,friedman1979multivariate}, the cross-match test \citep{rosenbaum2005exact}, and the energy test \citep{aslan2005new} reject the null hypothesis of no remodeling in all the three cases, in each of the 100 simulation replications (see Table \ref{tab:exp0} in Section \ref{sim:exp1}). This is not surprising, because these tests are designed to test the simple null hypothesis of equality of the two distributions. 
\begin{figure}[!h]
	\centering
	\includegraphics[width=0.95\linewidth]{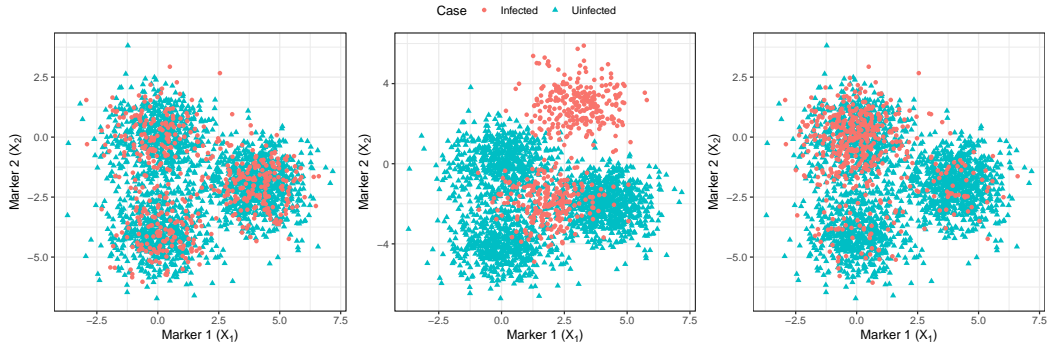}
	\caption{\small{Schematic representation of viral remodeling of infected cells versus preferential viral infection with respect to $d=2$ markers, $X_1$ and $X_2$. From left to right, we have (a) no remodeling, (b) remodeling, and  (c) no remodeling, but preferential infection. Uninfected cells are in green  whereas virus infected cells are in red.}\\[-6ex]}
	\label{fig:motfig}
\end{figure}


Due to the presence of subpopulation level heterogeneity the problem of testing for remodeling warrants testing a composite null hypothesis. To this end, note that under preferential infection, the two samples arise from the mixture distribution with identical component distributions but with different mixing weights. This is the case for the right most subplot in Figure \ref{fig:motfig}. In this paper, we formulate the problem of testing for preferential infection versus remodeling as a composite two-sample hypothesis with mixture distributions, and develop a new nearest-neighbor based test that can consistently and efficiently detect the differences between the two samples.

\subsection{The \texttt{TRUH} Testing Framework: Novel attributes and Our Contributions}

In this article, we propose a novel procedure for {\it Testing Remodeling Under Heterogeneity} (\texttt{TRUH}), that effectively incorporates the underlying heterogeneity and imbalance in the samples, and provides a conservative test for the composite null hypothesis that the two samples arise from the same mixture distribution but may differ with respect to the mixing weights. We summarize its key attributes below:

\begin{itemize}
	
	\item The \texttt{TRUH} statistic is based on a nearest neighbor approach \citep{cover1967nearest,devroye2013probabilistic} that relies on first identifying for every infected cell a predictive precursor cell which is the phenotypically closest cell in the uninfected population. It then measures the relative dissimilarities between the infected cells and their predictive precursors and the predictive precursors to their most phenotypically similar uninfected cells. A large relative dissimilarity between the infected cells and their predictive precursors indicates surface protein regulation or remodeling by the virus, while a small relative dissimilarity provides evidence for preferential infection or no remodeling. 
	
	
	\item We describe an efficient bootstrap based approach for calibrating the \texttt{TRUH} test statistic, and evaluate its performance in finite-sample simulations.  We then use this method to test for viral remodeling in tonsillar T cells under different types of HIV infection, corroborating the efficacy of our proposed procedure.
	
	\item  We provide an extensive theoretical understanding of the large sample characteristics of our proposed test statistics. We establish the $L_2$-limit of our proposed statistic using asymptotic properties of functionals of random geometric graphs \cite{py}. The limit can be expressed in terms of the densities of the uninfected and infected populations and dimension dependent constants obtained from nearest-neighbor distances defined on a homogeneous Poisson process. Using these properties, we can select a cut-off for the \texttt{TRUH} statistic that is asymptotically consistent against biologically-relevant location alternatives.  Traditional nonparametric tests enjoy these consistency properties in homogeneous populations but not under heterogeniety. We show that using a nearest neighbor based approach this inefficiency of existing nonparametric tests in heterogeneous data can be mitigated.


	
\end{itemize}

The rest of the paper is organized as follows: In Section \ref{sec:test}, we formulate the problem of testing for remodeling in single-cell virology as a heterogeneous two-sample problem, describe the \texttt{TRUH} framework, and show how it can be calibrated using the bootstrap. {Numerical experiments demonstrating the non-asymptotic performance of our testing procedure are given in Section \ref{sec:sims}. In Section \ref{sec:realdata} we use \texttt{TRUH} for studying remodeling in tonsillar T cells under different types of HIV infection.}
The asymptotic properties of the test statistic are discussed in Section \ref{sec:Tnm_limit}. 
We conclude the paper in Section \ref{sec:discuss} with a discussion. The technical details and proofs of the theoretical results are given in the supplementary materials. 

\section{Statistical Framework and the Proposed \texttt{TRUH} Statistic}
\label{sec:test}
In this section we formulate the problem of testing for remodeling in single-cell virology as a heterogeneous two-sample problem (Section \ref{sec:H0H1}), introduce the \texttt{TRUH} statistic (Section \ref{sec:Tnm}), and discuss how to calibrate it using the bootstrap (Section \ref{sec:bootstrap}). 

\subsection{The Heterogeneous Two-Sample Problem}
\label{sec:H0H1}
In our virology example, the baseline constitutes the $m$ uninfected cells. For each cell, $i \in \{1,\ldots m\}$, we denote by  $U_i$ a $d$-dimensional vector of cellular characteristics typically measuring expressions corresponding to different genes or proteins. Denote the uninfected/baseline population by $\bm U_m=\{U_1,\ldots,U_m\}$. Let $F_0$ be the cumulative distribution function (cdf) of the baseline population with
the heterogeneity in the population being reflected by $K$ different subgroups each having unimodal distributions with distinct modes and cdfs $F_{1}, \ldots, F_K$,  and mixing proportions $w_{1}, \ldots, w_K$, such that 
\begin{align}\label{eq:setup-1}
{F}_0 = \sum_{a=1}^K w_{a} \, F_{a}, \quad \text{ where } \quad w_{a} \in (0,1) \text{ and } \sum_{a=1}^K w_{a}=1.
\end{align}
Note that, the number of components $K$, the mixing distributions $F_1,\ldots, F_K$, and the mixing weights $w_1,\ldots, w_K$ are fixed (non-random) attributes, which are unknown. Also, as $F_{1}, \ldots, F_K$ are cdfs from unimodal distributions with distinct modes, ${F}_0$ is  well-defined with a unique specification. 
In addition to the uninfected population, we observe $n$ i.i.d. infected observations $\bm V_n=\{V_1,\ldots, V_n\}$ from  a distribution function $G$ in $\mbb{R}^d$.  Note that, the infected and uninfected samples $\bm U_m$ and $\bm V_n$ are collected from separate experiments and are independent of each other. 

\paragraph{Simple versus Composite Null}
In single-cell virology when an uninfected population is exposed to a pathogen, the virus may infect the different subpopulations at different rates. Therefore, even if the virus does not cause any change in the cellular characteristics, the virus infected sample might have different representations of the uninfected subpopulations than the uninfected mixing proportions $\{w_1, \ldots, w_K\}$. As such, it is quite possible that a few of uninfected subpopulations are completely absent in the viral population, which biologically implies that the virus preferentially targets few cellular subpopulations. Thus, if the virus does not induce any change in the cellular characteristics, then the distribution of the infected population $G$ lies in a class of distributions  $\mathcal{F}(F_0)$ that contains any convex combination of $\{F_1, \ldots, F_K\}$ including the boundaries, that is, 
\begin{align}
\label{eq:setup-2}
\mathcal{F}(F_0) = \left\{Q = \sum_{a=1}^K \lambda_a \, F_a: \lambda_1, \lambda_2, \ldots, \lambda_K \in [0,1] \text{ and } \sum_{a=1}^K \lambda_a=1 \right\}.
\end{align} 
Note that the uninfected cdf $F_0$ is a particular member of the class $\mathcal{F}(F_0)$. If the virus induces changes in the cellular characteristics, then the viral population distribution would contain at least one nontrivial subpopulation with distribution substantially different from $\{F_1, F_2, \ldots, F_K\}$ or their linear combinations. Thus, the test for viral remodeling  tantamounts to testing the following composite null hypothesis: 
\begin{equation}\label{eq:setup-3}
H_0: G \in \mathcal{F}(F_0) \quad \text{versus} \quad H_A: G \notin \mathcal{F}(F_0).
\end{equation}
If the null hypothesis is accepted, we say the virus exhibits {\it preferential infection}, otherwise we say the virus exhibits {\it remodeling} (see Figure \ref{fig:truh_working} below), and the hypothesis testing problem  \eqref{eq:setup-3} will be referred to as the problem of {\it testing remodeling under heterogeneity} (\texttt{TRUH}). Later on, to facilitate proofs of the theoretical properties of our proposed method, we will assume that the baseline cdfs $F_1, \ldots, F_K$ have unimodal densities $f_{1}, \ldots, f_K$ (with respect to Lebesgue measure). In this case, the baseline uninfected population will have density $f_0=\sum_{a=1}^K w_a f_a$, and the set of distributions in \eqref{eq:setup-2} can be represented in terms of the densities $f_{1}, \ldots, f_K$, and will be denoted by $\mc{F}(f_0)$.

\paragraph{Inefficiency of Existing Tests}
Traditional nonparametric graph-based two-sample tests, 
such as the edgecount (EC) test of \citet{friedman1979multivariate} or the crossmatch (CM) test of \citet{rosenbaum2005exact}, are tailored for the null hypothesis $H_0: F_0=G$, that is, testing whether the distributions of the uninfected samples  $\bm U_m$ and the infected samples $\bm V_n$ are the same. However, not surprisingly, direct application of these tests to the composite hypothesis testing problem described in \eqref{eq:setup-3} above gives non-conservative procedures. To see this, consider the EC test. Recall that the EC test is based on the statistic $\mc{R}(\bm U_m,\bm V_n)$ which counts the number of edges in the minimal spanning tree (MST)
of the pooled sample $\{U_1,\ldots,U_m, V_1,\ldots,V_n\}$ that connect points from different samples. Then, the EC test rejects the null hypothesis of $F_0=G$ for small values of $\mc{R}(\bm U_m,\bm V_n)$. The cut-off for $\mc{R}(\bm U_m,\bm V_n)$ can be chosen based on the asymptotic distribution $\mc{R}(\bm U_m,\bm V_n)$ under $F_0=G$, which was derived by \citet{henze1999multivariate} in the usual limiting regime where $m, n \to \infty$ and $n/m \to \rho \in (0,\infty)$. In particular, it follows from Theorem 1 of \citet{henze1999multivariate} that 
\begin{align}\label{eq:FG}
\lim_{m,n \to \infty} \mathbb{P}_{F_0=G}(\red{\mc{R}}(\bm U_m,\bm V_n)< C_{m,n}(\alpha)) = \alpha,
\end{align}
{with $C_{m,n}(\alpha)= \frac{2 mn}{m+n}-z_{1-\alpha} \sigma_{d} \sqrt{m+n}$, where $z_{1-\alpha}$ is the $\alpha$-th quantile of the standard normal distribution, $\sigma_d^2=\rho(4 \rho +(1-\rho)^2 \delta_d)/(1+\rho)^{4}$,  and $\delta_d$ is a constant depending only on dimension $d$.} 
More precisely, $\delta_d$ is the  variance of the degree of the origin $\bm 0\in \mbb{R}^d$ in the minimal spanning tree built on a homogeneous Poisson process of rate $1$ in $\mathbb{R}^d$ with the origin added to it. Note that \eqref{eq:FG} shows that the test with rejection region $\{\red{\mc R}(\bm U_m,\bm V_n)< C_{m,n}(\alpha)\}$ 
is asymptotically level $\alpha$ for the null hypothesis of $F_0=G$.


The following proposition shows that direct application of the EC test as described above,  will not be conservative for testing the hypothesis \eqref{eq:setup-3} of viral remodeling. 
In fact, for cases of preferential infection but no remodeling the EC test will produce undesired false discoveries. 

\begin{proposition}\label{INCONSISTENCY} Fix $\alpha \in (0,1/2)$. Then for $F_0$ as in \eqref{eq:setup-1} and for any $G \in \mathcal{F}(F_0)\setminus \{F_0\}$  in the usual limiting regime, 
	$$ \lim_{m,n \to \infty}\mathbb{P}\big(\red{\mc{R}}(\bm U_m,\bm V_n)< C_{m,n}(\alpha)\big)=1,$$
	with $\bm U_m=\{U_1,\ldots,U_m\}$ i.i.d. from $f_0$  and $\bm V_n=\{V_1,\ldots,V_m\}$ i.i.d. from $g$, where $f_0$ and $g$ are the densities (with respect to the Lebesgue measure) of $F_0$ and $G$, respectively. 
\end{proposition}  

The proof of the above result is given in the supplementary materials (Section A).  This shows that for any level $\alpha$, the EC test will be 
\red{inconsistent} as it would reject with certainty all cases of preferential infection but no remodeling. This phenomenon is demonstrated  in Figure~\ref{fig:ecfig} through a simple univariate simulation experiment.  Here, we consider $m=1000$, $n=50$, and $d=1$. 
\begin{figure}[!h]
	\centering
	\includegraphics[width=1\linewidth]{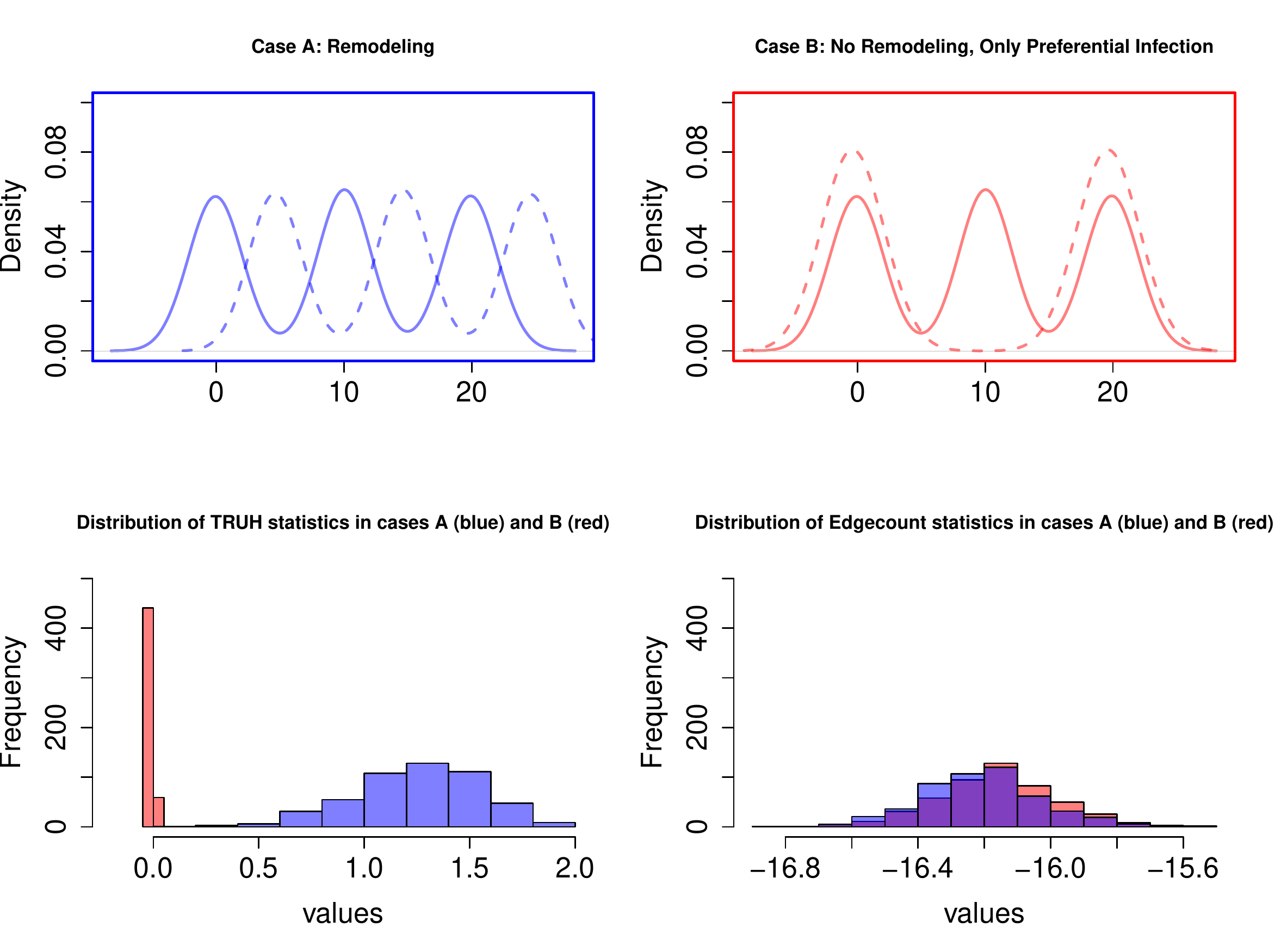}
	\caption{\small{Simulation example showing the performance of edgecount test statistic versus the \texttt{TRUH} statistic. In the top row, we describe the density of the true uninfected $F_0$ (in continuous line) and the density of the infected $G$ (in dotted line) for the two cases. In both cases, $F_0$ and $G$ are mixtures of normal distributions.
			In the first case, all the three equiprobable subpopulations in $F_0$ have undergone a discernible location change in $G$. In case B, $F_0$ again has three equiprobable subgroups while $G$ has two of those three subgroups. Thus, while case A signifies viral remodeling, there is no remodeling but only preferential infection in Case B. In the bottom row, we have the histogram of the values of the \texttt{TRUH} statistic in the left (defined below in \eqref{eq:TNM}) and the edgecount statistic in the right, respectively, under the two cases.}}
	\label{fig:ecfig}
\end{figure} 
The true distribution of the uninfected and infected subpopulations are gaussian mixtures. 
We consider two cases: 

\begin{itemize}
	
	\item {\it Case A}:  Here, $F_0$ and $G$ are \red{equal-weighted} mixtures of 3 
	Gaussians, with each subpopulation in $G$ having a different mean from those in $F_0$, that is, $F_0(u)=\frac{1}{3}\sum_{a=0}^2 \Phi(u-4 a)$ and $G(u)= \frac{1}{3}\sum_{a=0}^2 \Phi(u-4a-2)$.\footnote{Throughout, $\Phi(\cdot)$ and $\phi(\cdot)$ will denote the standard normal distribution function and density function, respectively.} This is a clear case of viral remodeling. 
	
	\item {\it Case B}: Here, $F_0=\frac{1}{3}\sum_{a=0}^2 \Phi(u-10 a)$ and $G=\frac{1}{2}\sum_{a=0}^1 \Phi(u-20 a)$. In this case, there is preferential infection, but no remodeling, that is, $G\in \mathcal{F}(F_0)$ with the middle population in $F_0$ being resistant to viral infection. 
	
\end{itemize}
Any test for the hypothesis \eqref{eq:setup-3} should ideally reject Case A and fail to reject Case B. However, Figure \ref{fig:ecfig} shows that the histogram of EC test statistic values across $500$ replications under cases A and B have a significant overlap. {Table \ref{tab:edgecount_comparison} shows  the rejection rate (proportion of false discoveries) in Case B and power (proportion of true discoveries) in case A as the level of the test is varied.}  From the table it is evident that there does not exists any choice of a critical value such that the rejection rate 
of the EC test in Case B is commendable as it rejects all cases of preferential infection presented under Case B. On the other hand, our proposed test statistic (\texttt{TRUH}), described in the following section, entertains possibilities where both the rejection rate and the power attain the desired limit. 

%
%
%
%
\begin{table}[!h]
	\centering
	\caption{ The rejection rate and the power of the edgecount and \texttt{TRUH} test statistics across $500$ repetitions of the simulation setting of Figure~\ref{fig:ecfig}.}
	\begin{tabular}{clrrrr}
		\hline
		& Level  & \multicolumn{1}{c}{0.01} & \multicolumn{1}{c}{0.05} & \multicolumn{1}{c}{0.10} & \multicolumn{1}{c}{0.20} \\
		\hline
		\multirow{2}[0]{*}{Power in Case A} & edgecount &  1.000     &  1.000     &  1.000     &  1.000\\
		& \texttt{TRUH}  &  1.000   &   1.000	&  1.000  & 1.000 \\
		\hline
		\multirow{2}[0]{*}{Rejection rate in Case B} & edgecount &   1.000    &  1.000     &  1.000  & 1.000 \\
		& \texttt{TRUH}  &    0.000   &   0.000	&  0.000  & 0.038 \\
		\hline
	\end{tabular}%
	\label{tab:edgecount_comparison}%
\end{table}%

\subsection{Proposed Test Statistic: \texttt{TRUH}}
\label{sec:Tnm}
In this section we describe a nearest-neighbor based statistic for testing the hypothesis of remodeling. To this end, recall that $\bm U_m=\{U_1, \ldots, U_m\}$ is the uninfected sample and $\bm V_n=\{V_1, \ldots, V_n\}$ is the infected sample.  Now, for each infected sample $V_i \in \bm V_n$, let 
\begin{align}\label{eq:D}
D_i=\min_{1 \leq j \leq m} ||V_i- U_j||,
\end{align}
the Euclidean distance of $V_i$ to its nearest point in the uninfected sample $\bm U_m$. The point in $\bm U_m$ which attains this minimum will be denoted by  $ N(V_i, \bm U_m)$\footnote{Given a finite set $S$ and any point $x \in \mbb{R}^d$, denote by  $ N(x, S)= \argmin_{y \in S} ||x-y||$,
	that is the nearest neighbor of $x$ in the set $S$. If there is a tie, that is, $N(x, S)$ has multiple elements, then we choose a random element from them and set that to $N(x, S)$.  However, if the underlying distribution of the data has a continuous density, then there are no ties with probability 1.} 
\begin{figure}[!t]
	\centering
	\includegraphics[width=4.85in]{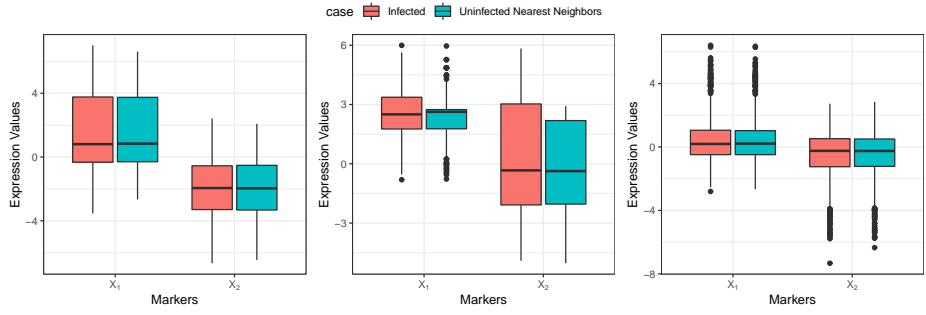}
	\caption{\small{Boxplots of the coordinates of $N_{\bm V_n,\bm U_m}=\{N(V_i, \bm U_m):1\le i\le n\}$ in green, adjacent to the boxplots of the coordinates of the corresponding infected cells $\bm V_n$ in red, for each of the scenarios discussed under Figure \ref{fig:motfig}. Recall that from left to right, we have (a) no remodeling, (b) remodeling, and  (c) no remodeling, but preferential infection.}}
	\label{fig:predprecursor}
\end{figure}
and constitutes a key \red{point in $\mbb{R}^d$ for}
measuring the relative phenotypic difference between the infected cells and their closest uninfected counterparts. In Figure \ref{fig:predprecursor} we show the boxplots of the coordinates of $N_{\bm V_n,\bm U_m}=\{N(V_i, \bm U_m):1\le i\le n\}$ in green, for each of the scenarios discussed under Figure \ref{fig:motfig}. Recall from Figure \ref{fig:motfig} that we have from left to right, (a) no remodeling, (b) remodeling, and  (c) no remodeling, but preferential infection. We note that for scenarios (a) and (c), the distributions of $N_{\bm V_n,\bm U_m}$ and $\bm V_n$ appear to overlap. However, in the case of remodeling (scenario (b) in the center plot), there is a clear difference between the two distributions for both the markers. The \texttt{TRUH} statistic captures this phenomenon and deals with the presence of heterogeneous groups (which can make the density within the  uninfected sample $\bm U_m$ to vary greatly), by comparing $D_i$ with a feature of the local density of $\bm U_m$ at $N(V_i, \bm U_m)$. For that purpose, define, for each infected observation, 
\begin{align}\label{eq:C}
C_i=\min_{1 \leq j \leq m: U_j \neq  N(V_i, \bm U_m)} || N(V_i, \bm U_m)-U_j||,
\end{align}
which is the distance of $ N(V_i, \bm U_m)$ to its nearest neighbor in $\bm U_m$. Our proposed test statistic for testing \eqref{eq:setup-3}, hereafter referred to as the \texttt{TRUH} statistic, is 
\begin{align}\label{eq:TNM}
T_{m, n}=\frac{1}{n^{1-\frac{1}{d}}}\left|\sum_{i=1}^n ( D_i-C_i)\right|= n^{\frac{1}{d}} |\bar D_{m, n}-\bar C_{m, n}|,
\end{align}
where $\bar D_{m, n}=\frac{1}{n} \sum_{i=1}^n D_i$ and $\bar C_{m, n}=\frac{1}{n} \sum_{i=1}^n C_i$. 
Note that the \texttt{TRUH} statistic above is an aggregated measure of how far apart each viral cell is from the uninfected sample compared to the local distance between uninfected sample points in its vicinity. Consider, for example, panel A in Figure \ref{fig:truh_working} that represents a schematic for remodeling, while panel B depicts preferential infection. Here, the three infected cells (in red) in Panel A are phenotypically different than their uninfected counterparts and thus the average gap $|\bar{D}_{m, n}-\bar{C}_{m, n}|$ in Panel A, averaged over the three infected cells, is relatively larger than what is observed under preferential infection in Panel B.
\begin{figure}[!h]
	\centering
	\includegraphics[width=0.8\linewidth]{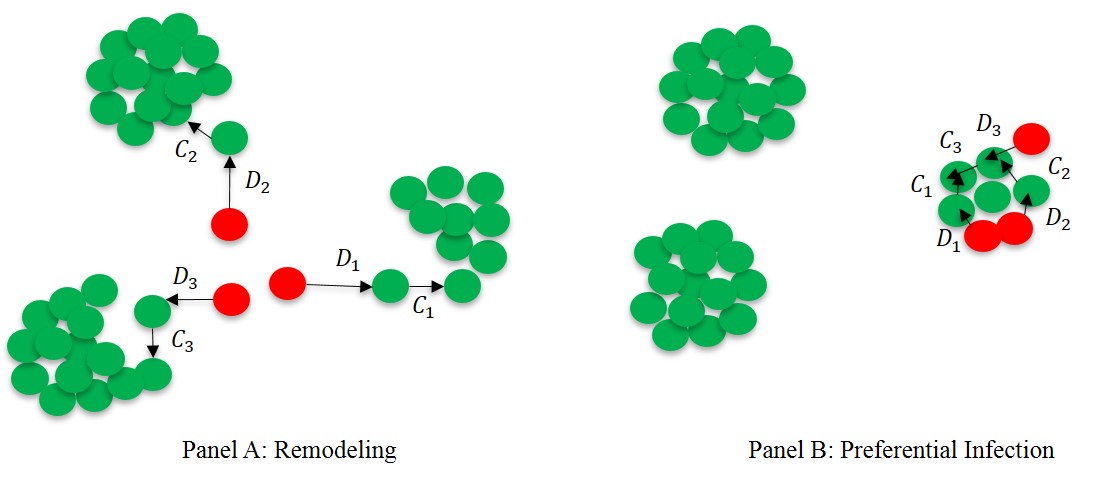}
	\caption{\small{Panel A represents the scenario of remodeling while Panel B exhibits Preferential Infection. Uninfected cells are in green while infected cells are in red. The gaps   are larger in case of remodeling as infected cells are phenotypically different than their uninfected counterparts.}}
	\label{fig:truh_working}
\end{figure}
Therefore, we develop a test to reject the null hypothesis of no remodeling for large values of $T_{m,n}$. The cut-off for $T_{m,n}$ can be chosen based on a bootstrap calibration (Section \ref{sec:bootstrap}) or using the asymptotic limit of $T_{m, n}$ (Section \ref{sec:Tnm_limit}). Note that, since the nearest neighbor of a point, in a cloud of $n$ random points in $\mbb{R}^d$, typical lies within a ball of radius $n^{-\frac{1}{d}}$ centered at that point, the \texttt{TRUH} statistic is scaled by $n^{1-\frac{1}{d}}$, which makes $T_{m,n}$ bounded in probability.

One of the interesting properties of the quantity $T_{m, n}$ is that it only involves enumeration of distance based features for the viral sample, unlike classical graph-based two-sample tests  \citep{rosenbaum2005exact,friedman1979multivariate} which are built using the inter-point distances of the pooled sample. As a consequence, the \texttt{TRUH} test statistic is not symmetric in its usage of the uninfected and infected samples, even  when the sample sizes are equal and the two samples were actually generated from the same population distribution. This asymmetric sample usage of \texttt{TRUH} helps in tackling possibly different heterogeneity levels in the two samples. Finally, note that even though the quantities $D_i$ and $C_i$ are defined above using the Euclidean distance, they can be easily generalized to any arbitrary distance function, and the statistic $T_{m, n}$ can potentially be used in non-Euclidean data spaces, such as graph data or functional data, as well. 
\subsection{Bootstrap Based Calibration for \texttt{TRUH}}
\label{sec:bootstrap}

In this section, we present a bootstrap based procedure to determine the cut-off 
$t_{m,n, \alpha}$ for a level $\alpha$ test using $T_{m,n}$. To this end, recall that $\mc{F}(F_0)$ contains any convex combination of the baseline distribution functions $\{F_1, \ldots, F_K\}$. Therefore, the proposed bootstrap procedure relies \red{on the following two steps: (i) random sampling of the mixing proportions a large number of times, and (ii) for each such sampled mixing proportion, surrogate samples from $\mc{F}(F_0)$ are constructed to generate a pseudo null distribution which is used to estimate the level $\alpha$ cut-off. The maximum of all the level $\alpha$ cut-offs so obtained, one for each sampled mixing proportion, is then used to calibrate the \texttt{TRUH} statistic.}

Our algorithm leverages the fact that in our virology example the number $m$ of uninfected samples is much larger than the size $n$ of the infected samples. Therefore,  
we can use the prediction strength approach of \citet{tibshirani2005cluster} on the uninfected samples to obtain an estimate $\hat K$ of the unknown number of heterogeneous subgroups $K$. We then use this value of $\hat K$ to estimate the class memberships of the baseline samples $\bm U_m$ using a $\hat K$-means algorithm. For $1 \leq a \leq \hat K$, denote by $\hat J_a \subseteq \{1, 2, \ldots, m\}$ the subset of indices which belong to class $a$ in the output of the $\hat K$-means algorithm. Let $\bm U_{\hat J_a}=\{U_i: i \in \hat J_a \}$ be the subset of the baseline samples estimated to be in the $a^{th}$ class by the $\hat K$-means algorithm. Note that $\bm{U}_m=\{\bm{U}_{\hat J_a}: a=1, 2, \ldots,\hat{K}\}$ and $\sum_{a=1}^{\hat{K}} m_a=m$, where $m_a=|\hat J_a|$.  

Now, \red{for each $b_1=1,\ldots,B_1$, denote by $(\lambda_1^{(b_1)}, \ldots, \lambda_{\hat{K}}^{(b_1)})$ a random sample from the $\hat K$-dimensional simplex $\mc{S}_{\hat K}=\{(z_1, \ldots, z_{\hat K}) \in \mbb{R}^{\hat K}: z_a \in [0, 1], \text{ for } 1 \leq a \leq \hat{K}, \text{ and } \sum_{a=1}^{\hat K} z_a=1\}$. Given the mixing weights $\{\lambda_1^{(b_1)},$ $\ldots, \lambda_{\hat K}^{(b_1)}\}$,  we construct $B_2$ surrogate infected samples from $\mc{F}(F_0)$ as follows: for each $b_2=1,\ldots,B_2$ and for $1 \leq a \leq \hat{K}$,  randomly sample $\lceil n\lambda_a^{(b_1)} \rceil$ elements without replacement from $\bm{U}_{\hat J_a}$. Denote the chosen elements by $$\mc{V}^{(b_2)}_a=\{U^{(b_2)}_1,\ldots, U^{(b_2)}_{\lceil n\lambda_a^{(b_1)} \rceil}\},$$ and set the remaining $m_a-\lceil n\lambda_a^{(b_1)} \rceil$ elements in $\bm{U}_{\hat J_a}$  as the residual baseline sample $\mc{U}^{(b_2)}_{a}$ in class $a$. Now, combining the samples over the $\hat K$ classes, we get the surrogate infected sample as $\bm{V}^{(b_2)}_{n}=\{\mc{V}^{(b_2)}_{a}: a=1,\ldots,\hat{K}\}$ and the corresponding baseline sample as $\bm{U}^{(b_2)}_{\tilde{m}}=\{\mc{U}^{(b_2)}_{a}: a=1,\ldots,\hat{K}\}$, where $$\tilde{m}=\sum_{a=1}^{\hat{K}}(m_a-\lceil n\lambda^{(b_1)}_a \rceil).$$ }
Note that under the null hypothesis of no remodeling ($G \in \mc{F}(F_0)$),  the bootstrapped samples in the \red{$b_2^{th}$} round, $\bm{U}^{(b_2)}_{\tilde{m}}$ and $\bm{V}^{(b_2)}_{n}$ (which are surrogates for $\bm U_m$ and $\bm V_n$, respectively), can be used to compute the statistic 
\begin{align}\label{eq:TNM_II}
\red{T_{\tilde{m}, n}^{(b_2)}=n^{\frac{1}{d}} |\tau_{fc} \cdot \bar D_{\tilde{m}, n}-\bar C_{\tilde{m}, n}|.}
\end{align}
\red{For $b_1$ fixed, $T_{\tilde{m}, n}^{(b_2)}$ is the surrogate of the \texttt{TRUH} statistic in the \red{$b_2^{th}$} bootstrap round.}
Observe that compared to \eqref{eq:TNM}, we have introduced a tuning parameter $\tau_{fc}$ in \eqref{eq:TNM_II} above. 
We define it as the fold change (fc) hyper-parameter and will consider values of $\tau_{fc}\geq 1$. 
Biologically relevant remodeling corresponds to significant fold change increase or decrease in the magnitude of cellular expressions between the infected and the uninfected cells. 
As we test the global null hypothesis of no change in any of the concerned genes, alternative hypothesis of remodeling with meager fold changes, if accepted, will only lead to biologically uninteresting discoveries. For discovering virologically interesting alternatives,  it is natural to set $\tau_{fc}$ slightly larger than 1. (Note that $\tau_{fc}=1$ corresponds to the bootstrapped version of the \texttt{TRUH} statistic in \eqref{eq:TNM}.) In the simulation experiments presented later in Section~\ref{sec:sims} we set $\tau_{fc}=1$ whereas in Section \ref{sec:realdata} $\tau_{fc}$ is fixed at $1.1$ as we study a real-world virology dataset.

\begin{algorithm}[!h]
	\KwIn{The parameters $n$, $\tau_{fc}$, and $\alpha$. The baseline sample $\bm{U}_m$, and the estimates $\hat{K}$ and $\{\hat J_a : a=1,\ldots,\hat{K}\}$ from the $K$-means algorithm.}
	\KwOut{The bootstrapped level $\alpha$ cutoff $t_{m, n, \alpha}$.}
	\For{$b_1=1,\ldots,B_1$}{
		STEP 1: Random sample $\{\lambda_1^{(b_1)}, \ldots, \lambda_{\hat K}^{(b_1)}\}$ from the $\hat K$-dimensional simplex\;
		\For{$b_2=1,\ldots,B_2$}{
			\For{$a=1,\ldots,\hat{K}$}{
				
				\eIf{$\lceil n\lambda_a^{(b_1)} \rceil \le{m}_a$}{
					STEP 2: Draw a simple random sample 
					$\mc{V}_a^{(b_2)}=\{U^{(b_2)}_1, \ldots, U^{(b_2)}_{\lceil n\lambda_a^{(b_1)} \rceil}\}$ without replacement from $\bm{U}_{\hat J_a}$\;
					STEP 3: $\mc{U}_a^{(b_2)}=\bm{U}_{\hat J_a} \backslash \mc{V}_a^{(b_2)}$ the baseline residual sample in class $a$\;
					
				}{
					Stop: Go to STEP 1\;
					
				}
			}
			Surrogate Case sample: $\bm{V}^{(b_2)}_{n}=\{\mc{V}^{(b_2)}_{a}: a=1,\ldots,\hat{K}\}$\;
			Baseline sample: $\bm{U}^{(b_2)}_{\tilde{m}}=\{\mc{U}^{(b_2)}_{a}: a=1,\ldots,\hat{K}\}$\;
			STEP 4: Calculate $T_{\tilde{m}, n}^{(b_2)}=n^{\frac{1}{d}} |\tau_{fc}~\bar D_{\tilde{m}, n}-\bar C_{\tilde{m}, n}|$\;
		}
		STEP 5: Return $t_{m,n, \alpha}^{(b_1)}=\min\{T_{\tilde{m}, n}^{(b_2)}: \frac{1}{B_2} \sum_{r=1}^{B_2}\bm 1 \{T_{\tilde{m}, n}^{(r)}\ge T_{\tilde{m}, n}^{(b_2)}\} \le \alpha\}$.
	}
	STEP 6: Return $t_{m,n, \alpha}=\max\{t_{m,n, \alpha}^{(b_1)}: 1\le b_1\le B_1\}$.
	\caption{Bootstrap cut-off for a level $\alpha$ test using $T_{m, n}$}
	\label{algo:max}
\end{algorithm}

The bootstrap procedure described above is summarized in Algorithm \ref{algo:max}. \red{The computational complexity of Algorithm \ref{algo:max} is driven by the following two steps: (i) the computation of the estimated number of clusters $\hat{K}$, and (ii) the computation of the \texttt{TRUH} test statistic over the bootstrap samples.
	While the calculations in step (ii) can be distributed across the $B_1B_2$ bootstrap samples and $n$ infected samples,} the computational cost of estimating $T_{\tilde{m}, n}^{(b)}$ for a fixed $b$ is \red{$O(md)$ which is the cost of running the 1-nearest neighbor algorithm twice for each of the $n$ infected samples.} 
To estimate $K$, we use prediction strength along with a $K$-means algorithm where the target number of clusters and the maximum number of iterations over which the $K$-means algorithm runs before stopping are both fixed and thus has $O(md)$ complexity. Therefore the overall computational complexity of Algorithm \ref{algo:max} is \red{$O(md)$}. \red{For the numerical experiments and real data analysis of Sections \ref{sec:sims} and \ref{sec:realdata}, we set $B_2=200$ and implement a version of Algorithm \ref{algo:max} which samples the mixing proportions $\{\lambda_1, \ldots, \lambda_{\hat K}\}$ only from the corners of the $\hat{K}$ dimensional simplex $\mc{S}_{\hat{K}}$ as follows: we set $B_1=\hat{K}$ and for $b_1=1,\ldots, B_1$, and $a=1,\ldots,\hat{K}$, we take $\lambda_a^{(b_1)}=1$ if $b_1=a$ and $0$ otherwise. This sampling scheme ensures that the mechanism for generating the mixing proportions places most weight on the corners of $\mc{S}_{\hat{K}}$.} 

\section{Numerical Experiments}
\label{sec:sims}
In this section we evaluate the numerical performance of the \texttt{TRUH} procedure across a wide range of simulation experiments. We consider the following six competing testing procedures that use different methodologies to conduct a nonparametric two-sample test: (i) Energy test (\texttt{Energy}) of \citet{aslan2005new} available from the R package \texttt{energy}, (ii) Cross-Match test (\texttt{Crossmatch}) of \citet{rosenbaum2005exact} available from the R package \texttt{crossmatch}, (iii) edgecount test (\texttt{E Count}) of \citet{friedman1979multivariate}, (iv) Generalized edgecount test (\texttt{GE Count}) of \citet{chen2017new}, (v) Weighted edgecount test (\texttt{WE Count}) of \citet{chen2018weighted}, and (vi) the Max Type edgecount test (\texttt{MTE Count}) of \citet{zhangmaxtypec}. The aforementioned four edge count based tests are available from the R package \texttt{gtests}. \red{We note that the preceding six testing procedures are not designed to test the composite null hypothesis of equation \eqref{eq:setup-3} and rely on a simple null hypothesis $H_0:F_0=G$ for inference. Nevertheless, the simulation experiments presented in this section highlight the incorrect inference that may result when traditional two-sample tests are used for testing the composite null hypothesis of equation \eqref{eq:setup-3}.}

To assess the performance of the competing testing procedures, we simulate $\bm{U}_m$ and $\bm V_n$ from $F_0$ and $G$, the cdf of the baseline and the infected population respectively, and for each testing procedure, we measure the proportion of rejections across $100$ repetitions of the composite null hypothesis test described in \eqref{eq:setup-3} at $5\%$ level of significance. For \texttt{TRUH}, we \red{use Algorithm \ref{algo:max} with fold change constant $\tau_{fc}=1$, $B_2=200$ and sample the mixing proportions only from the corners of $\hat{K}$ dimensional simplex $\mc{S}_{\hat{K}}$ as described in section \ref{sec:bootstrap}}. The R code that reproduces our simulation results can be downloaded from the following link: \href{https://github.com/trambakbanerjee/TRUH_paper}{https://github.com/trambakbanerjee/TRUH\_paper} and the associated R package is available at \href{https://github.com/trambakbanerjee/TRUH#truh}{https://github.com/trambakbanerjee/TRUH}.

\subsection{Experiment 1}
\label{sim:exp1}

In the setup of experiment 1, we consider testing $H_0: G \in \mathcal{F}(F_0) \quad \text{versus} \quad H_A: G \notin \mathcal{F}(F_0)$, when $F_0$ is the cdf of a $d$ dimensional Gaussian mixture distribution with three components:
$$F_0=0.3\mc{N}_d(\bm\mu_1,\bm \Sigma_1)+0.3\mc{N}_d(\bm\mu_2,\bm \Sigma_2)+0.4\mc{N}_d(\bm\mu_3,\bm \Sigma_3),
$$
where $\bm \mu_1=\bm 0_{d},~\bm \mu_2=-3\bm 1_{d},~\bm\mu_3=-\bm{\mu}_2$, and $\bm \Sigma_K$, for $K=1,2,3$, are $d$ dimensional positive definite matrices with eigenvalues randomly generated from the interval $[1,~10]$. To simulate $\bm V_n$ from $G$, we consider two scenarios as follows:
\begin{itemize}
	
	\item Scenario I: Here, $G=0.1\,\mc{N}_d(\bm\mu_1,\bm \Sigma_1)+0.1\,\mc{N}_d(\bm\mu_2,\bm \Sigma_2)+0.8\,\mc{N}_d(\bm\mu_3,\bm \Sigma_3)$. In this case, $G$ has all the subpopulations present in $F_0$ but at different proportions. Thus, $G\in\mc{F}(F_0)$, and the correct inference here is no remodeling. 
	
	\item Scenario II: This setting presents a scenario where $G\notin\mc{F}(F_0)$ and the composite null $H_0$ is not true. Here, we consider $G=0.5\mc{N}_d(\bm\mu_1,\bm \Sigma_1)+0.5\mc{N}_d(\bm\mu_4,\bm \Sigma_4)$, where $\bm \Sigma_4$ is a $d$ dimensional positive definite matrix generated independently of $\bm \Sigma_1,\bm \Sigma_2,\bm \Sigma_3$, and $\bm \mu_4=4\bm \epsilon_d$, where $\bm \epsilon_d$ \red{is a vector of $d$ independent} 
	Rademacher random variables. 
\end{itemize}
\begin{table}[!h]
	\centering
	\caption{Rejection rates at $5\%$ level of significance: Experiment 1 and Scenario I wherein $H_0: G \in \mathcal{F}(F_0)$ is true.}
	\begin{tabular}{lcccccc}
		\hline
		& \multicolumn{3}{c}{$m = 500,~n = 50$} & \multicolumn{3}{c}{$m = 2000,~n = 200$} \\
		\hline
		\multicolumn{1}{c}{Method} & $d=5$   & $d=15$  & $d=30$  & $d=5$   & $d=15$  & $d=30$ \\
		\hline
		\texttt{Energy} & 1.000  & 1.000  & 1.000  & 1.000    & 1.000     & 1.000 \\
		\texttt{Crossmatch} & 0.220  & 0.150  & 0.145  & 0.460  & 0.410  & 0.340 \\
		\texttt{E Count} & 0.185  & 0.115  & 0.055  & 0.400   & 0.335 & 0.195 \\
		\texttt{GE Count} & 0.170  & 0.185  & 0.225  & 0.510  & 0.540  & 0.605 \\
		\texttt{WE Count} & 0.300  & 0.295  & 0.360  & 0.655 & 0.745 & 0.735 \\
		\texttt{MTE Count} & 0.230  & 0.230  & 0.290  & 0.605 & 0.665 & 0.665 \\
		\hline
		\texttt{TRUH}  & 0.02  & 0.015  & 0.015  & 0.01 & 0.02 & 0.01 \\
		\hline
	\end{tabular}%
	\label{tab:exp1set1}%
\end{table}%
For scenario I, table \ref{tab:exp1set1} reports the rejection rates for $100$ repetitions of the test for varying $d,m,n$ when the parameters $\{\bm \mu_i,\bm \Sigma_i\,1\le i\le 4\}$ are held fixed across these repetitions. We see that \texttt{TRUH} returns the smallest rejection rate.  The other six tests all have very high rejection rates \red{as they fail to account for the composite nature of the null hypothesis.}  
The rejection rate for \texttt{TRUH} is below the \red{prespecified} 
$0.05$ level establishing that it is a conservative test across all the regimes considered in the table. In scenario II, however, we find that all the tests correctly identify $G\notin \mc{F}(F_0)$ in all the regimes and across all replications. This shows, all the tests exhibit perfect rejection rates in this scenario. These two scenarios under experiment 1 demonstrate that for testing the composite null hypothesis of equation \eqref{eq:setup-3}, direct application of traditional two-sample tests such as those considered here, is no longer conservative \red{as these tests rely on a simple null hypothesis for inference.} 
\texttt{TRUH}, \red{on the other hand, is adept at detecting $H_0:G\in\mc{F}(F_0)$ and powerful against departures from $H_0$.}
\begin{table}[!h]
	\small
	\centering
	\caption{Rejection rates at $5\%$ level of significance: Simulation experiment corresponding to Figure \ref{fig:motfig}.}
	\scalebox{0.75}{
		\begin{tabular}{lccc}
			\hline
			& \multicolumn{3}{c}{$m = 2000,~n = 500,~d=2$}\\
			\hline
			\multicolumn{1}{c}{Method} & Left panel:  & Center panel: & Right panel:\\
			\multicolumn{1}{c}{} & no remodeling $(G\in\mc{F}(F_0))$ & remodeling $(G\notin\mc{F}(F_0))$  & preferential infection $(G\in\mc{F}(F_0))$\\
			\hline
			\texttt{Energy} & 0.030 & 1.000 & 1.000 \\
			\texttt{Crossmatch} &0.030 & 1.000 & 1.000 \\
			\texttt{E Count} & 0.010 & 1.000 & 1.000 \\
			\texttt{GE Count} & 0.000 & 1.000 & 1.000\\
			\texttt{WE Count} & 0.060 & 1.000 & 1.000\\
			\texttt{MTE Count} & 0.030 & 1.000 & 1.000 \\
			\hline
			\texttt{TRUH}  & 0.000 & 0.980 & 0.000\\
			\hline
	\end{tabular}}%
	\label{tab:exp0}%
\end{table}%

\normalsize

In Table \ref{tab:exp0} we present the results of the simulation exercise that correspond to the three scenarios described in figure \ref{fig:motfig}. The two dimensional uninfected marker expressions $(X_1,X_2)$ are randomly sampled from $F_0=w_1\mc{N}_2(\bm \mu_1,\bm I_2)+w_2\mc{N}_2(\bm \mu_2,\bm I_2)+w_3\mc{N}_2(\bm \mu_3,\bm I_2)$, where $\bm \mu_1=\bm 0,~\bm \mu_2=(0,-4),~\bm \mu_3=(4,-2)$ and the sample size is $m=2000$. The mixing weights are given by $(w_1,w_2,w_3)=(0.3,0.3,0.4)$. For the panel on the left of figure \ref{fig:motfig}, infected marker expressions arise from $F_0$ but with sample size $n=500$, while for the center panel the infected marker expressions represent a random sample of size $n$ from $G=0.5 \mc{N}_2(\bm \mu_4,\bm I_2)+0.5\mc{N}_2(\bm \mu_5,\bm I_2)$, where $\bm \mu_4=0.25 \bm \mu_2+0.5\bm \mu_3$ and $\bm \mu_5=(3/4) \bm \mu_2+(9/8)\bm \mu_3$. Clearly in this case $G\notin\mc{F}(F_0)$. For the right most panel, infected marker expressions are again a random sample of size $n$ from $\mc{F}(F_0)$ with mixing weights given by the vector $(w_1,w_2,w_3)=(0.8,0.1,0.1)$. Under this setting, the three \texttt{g-tests} \citep{chen2018weighted,chen2017new,friedman1979multivariate}, the cross-match test of \citet{rosenbaum2005exact}, and the energy test of \citet{aslan2005new} infer $G\notin\mc{F}(F_0)$ in all of the $100$ repetitions of the experiment thus suggesting their inability to tackle subpopulation level heterogeneity.
\subsection{Experiment 2}
\label{sim:exp2}
For experiment 2, we consider a more complex setup wherein $F_0$ is the cdf of a $d$ dimensional mixture distribution which is not necessarily Gaussian. Here, 
$$F_0=0.5~{\mathrm{Gam}}_d(\text{shape}=5\bm 1_d,\text{rate}=\bm 1_d,\bm \Sigma_1)+0.5~{\mathrm{Exp}}_d(\text{rate}=\bm 1_d,\bm \Sigma_2),$$ 
where ${\mathrm{Gam}}_d$ and ${\mathrm{Exp}}_d$ are $d$ dimensional Gamma and Exponential distributions. For generating correlated Gamma and Exponential variables, we use the Gaussian copula approach based function from the R-package \texttt{lcmix} \citep{lcmix,xue2000multivariate}. 
We consider tapering matrices with positive and negative autocorrelations: $(\bm \Sigma_1)_{ij}=0.7^{|i-j|}$ and $(\bm \Sigma_2)_{ij}=-0.9^{|i-j|}$ for $1\le i,j\le d$.
For simulating $\bm V_n$ from $G$, we consider the following two scenarios:
\begin{itemize}
	
	\item Scenario I: Here, $G={\mathrm{Exp}}_d(\text{rate}=\bm 1_d,\bm \Sigma_2)$. In this case, $G$ arises from only one of the components of $F_0$, that is, $G\in\mc{F}(F_0)$.
	
	\item Scenario II: Here, $G=0.1~{\mathrm{Gam}}_d(\text{shape}=10\bm 1_d,\text{rate}=0.5\bm 1_d,\bm \Sigma_1)+0.9~{\mathrm{Exp}}_d(\text{rate}=\bm 1_d,\bm \Sigma_2)$. In this setting, $G\notin\mc{F}(F_0)$ and the composite null $H_0$ is not true. When the ratio $n/m$ is small, this scenario presents a difficult setting for detecting departures from $H_0$ as majority of the case samples from $\bm V_n$ will arise from ${\mathrm{Exp}}_d(\text{rate}=\bm 1_d,\bm \Sigma_2)$ and the tests will rely on only a small fraction of samples from ${\mathrm{Gam}}_d(\text{shape}=10\bm 1_d,\text{rate}=0.5\bm 1_d,\bm \Sigma_1)$ to reject the null hypothesis. 
\end{itemize}
\begin{table}[!h]
	\centering
	\caption{Rejection rates at $5\%$ level of significance: Experiment 2 and Scenario I wherein $H_0: G \in \mathcal{F}(F_0)$ is true.}
	\begin{tabular}{lcccccc}
		\hline
		& \multicolumn{3}{c}{$m = 500,~n = 50$} & \multicolumn{3}{c}{$m = 2000,~n = 200$} \\
		\hline
		\multicolumn{1}{c}{Method} & $d=5$   & $d=15$  & $d=30$  & $d=5$   & $d=15$  & $d=30$ \\
		\hline
		\texttt{Energy} & 1.000 & 1.000 & 1.000 & 1.000 & 1.000 & 1.000 \\
		\texttt{Crossmatch} & 0.460 & 0.440 & 0.390 & 0.800 & 0.850 & 0.760 \\
		\texttt{E Count} & 0.290 & 0.190 & 0.280 & 0.720 & 0.690 & 0.560 \\
		\texttt{GE Count} & 0.400 & 0.430 & 0.390 & 0.900 & 0.920 & 0.900 \\
		\texttt{WE Count} & 0.560 & 0.590 & 0.600 & 0.970 & 0.960 & 0.940 \\
		\texttt{MTE Count} & 0.460 & 0.510 & 0.440 & 0.930 & 0.950 & 0.910 \\
		\hline
		\texttt{TRUH}  & 0.000 & 0.000 & 0.000 & 0.000 & 0.000 & 0.000 \\
		\hline
	\end{tabular}%
	\label{tab:exp2set1}%
\end{table}%
\begin{table}[htbp]
	\centering
	\caption{Rejection rates at $5\%$ level of significance: Experiment 2 and Scenario II wherein $H_0: G \in \mathcal{F}(F_0)$ is false.}
	\begin{tabular}{lcccccc}
		\hline
		& \multicolumn{3}{c}{$m = 500, n = 10$} & \multicolumn{3}{c}{$m = 2000, n = 40$} \\
		\hline
		\multicolumn{1}{c}{Method} & $d=5$   & $d=15$  & $d=30$  & $d=5$   & $d=15$  & $d=30$\\
		\hline
		\texttt{Energy} & 0.930 & 0.960 & 1.000 & 1.000 & 1.000 & 1.000 \\
		\texttt{Crossmatch} & 0.400 & 0.350 & 0.470 & 0.600 & 0.720 & 0.720 \\
		\texttt{E Count} & 0.180 & 0.120 & 0.130 & 0.340 & 0.310 & 0.200 \\
		\texttt{GE Count} & 0.310 & 0.230 & 0.160 & 0.800 & 0.790 & 0.770 \\
		\texttt{WE Count} & 0.510 & 0.490 & 0.460 & 0.800 & 0.790 & 0.790 \\
		\texttt{MTE Count} & 0.390 & 0.430 & 0.380 & 0.800 & 0.780 & 0.770 \\
		\hline
		\texttt{TRUH}  & 0.580 & 0.580 & 0.580 & 0.880 & 0.940 & 0.960 \\
		\hline
	\end{tabular}%
	\label{tab:exp2set2}%
\end{table}%
Table \ref{tab:exp2set1} reports the rejection rates, for $100$ repetitions, of the different tests in  scenario I. Note that \texttt{TRUH} correctly identifies that $G\in\mc{F}(F_0)$ while the remaining tests overwhelmingly support $G\notin\mc{F}(F_0)$, especially when $m$ is large, demonstrating their lack of conservatism in testing the composite null hypothesis of the form \eqref{eq:setup-3}. 
The results for scenario II (Table \ref{tab:exp2set2}) are reported for $n/m=0.02$, where, with the exception of \texttt{Energy} test, all the other competing tests demonstrate small rejection rates for $m=500$. Substantial improvement in the rejection rates is evident when $m=2000$. However, for both these cases, $m=500$ and $m=2000$, the \texttt{Energy} test followed by  \texttt{TRUH} exhibit the largest rejection rates. Although \texttt{Energy} test rejects $H_0$ in almost all of the testing instances in scenario II, its performance in scenario I (Table \ref{tab:exp2set1}) reveals that it can be severely non-conservative when \red{testing under a composite null hypothesis $H_0:G \in \mathcal{F}(F_0)$.} 

\subsection{Experiment 3}
\label{sim:exp3}
For experiment 3, we introduce zero inflation in both the baseline and case samples to mimic the scenario that is often encountered in virology studies wherein some of the markers exhibit only a small probability of expressing themselves. We let $\bm p=(p_1,\ldots,p_d)$ denote the $d$ dimensional vector of point masses at $0$ across dimensions, and consider $$F_0= 0.5 F_1 + 0.5 F_2,$$ where $F_1=\bm p\delta_{\{\bm 0\}}+(\bm 1_d-\bm p)~{\mathrm{Gam}}_d(\text{shape}=5\bm 1_d,\text{rate}=\bm 1_d,\bm \Sigma_1)$ and $F_2=\bm p\delta_{\{\bm 0\}}+(\bm 1_d-\bm p)~{\mathrm{Exp}}_d(\text{rate}=\bm 1_d,\bm \Sigma_2)$. In the above representation, $\bm p$ regulates the differential zero inflation across the $d$ dimensions. For the purposes of this experiment, we chose the first $0.8d$ coordinates of $\bm p$  independently from $\text{Unif}(0.5, 0.6)$, and the remaining $0.2d$ coordinates are set to $0$. Thus, the zero inflation is encountered only in the first $0.8d$ coordinates of $F_0$. To simulate the baseline sample $\bm{U}_m$ from $F_0$, we use the R-package \texttt{lcmix} with $\bm \Sigma_1, \bm \Sigma_2$ as described in experiment 2 (Section \ref{sim:exp2}). For simulating $\bm V_n$ from $G$, we consider the following two scenarios:
\begin{itemize}
	\item Scenario I: Let $G=\bm p\delta_{\{\bm 0\}}+(\bm 1_d-\bm p)~{\mathrm{Exp}}_d(\text{rate}=\bm 1_d,\bm \Sigma_2)$. Here, $G$ arises from only one of the components of $F_0$, that is, $G\in\mc{F}(F_0)$.
	\item Scenario II: Here, we let $G=0.5 G_1+0.5 G_2$, where $G_1=\bm q\delta_{\{\bm 0\}}+(\bm 1_d-\bm q)~{\mathrm{Gam}}_d(\text{shape}=5\bm 1_d,\text{rate}=0.5\bm 1_d,\bm \Sigma_1)$ 
	and $G_2=\bm q\delta_{\{\bm 0\}}+(\bm 1_d-\bm q)~{\mathrm{Exp}}_d(\text{rate}=\bm 1_d,\bm \Sigma_2)$, and we set the first $0.8d$ coordinates of $\bm q$ to $0.3$ and the remaining $0.2d$ coordinates to $0$. Note that in this setting, apart from the difference in the rate parameter of the Gamma distribution, we also have differential zero inflation across $G$ and $F_0$, as $\bm q\ne \bm p$. Thus, $G\notin\mc{F}(F_0)$ and the composite null $H_0$ is not true. Moreover, when $n$ is small, this scenario presents a challenging setting for detecting departures from $H_0$ as the tests will have to rely on both the differences in the rate parameter and differential zero expression between $\bm{U}_m, \bm V_n$ to reject the null hypothesis. 
\end{itemize}
\begin{table}[!h]
	\centering
	\caption{Rejection rates at $5\%$ level of significance: Experiment 3 and Scenario I wherein $H_0: G \in \mathcal{F}(F_0)$ is true.}
	\begin{tabular}{lcccccc}
		\hline
		& \multicolumn{3}{c}{$m = 500,~n = 50$} & \multicolumn{3}{c}{$m = 2000,~n = 200$} \\
		\hline
		\multicolumn{1}{c}{Method} & $d=5$   & $d=15$  & $d=30$  & $d=5$   & $d=15$  & $d=30$ \\
		\hline
		\texttt{Energy} & 1.000 & 1.000 & 1.000 & 1.000 & 1.000 & 1.000 \\
		\texttt{Crossmatch} & 0.340 & 0.330 & 0.240 & 0.730 & 0.800 & 0.670 \\
		\texttt{E Count} & 0.200 & 0.120 & 0.100 & 0.670 & 0.440 & 0.340 \\
		\texttt{GE Count} & 0.300 & 0.290 & 0.330 & 0.870 & 0.860 & 0.870 \\
		\texttt{WE Count} & 0.510 & 0.460 & 0.540 & 0.970 & 0.920 & 0.930 \\
		\texttt{MTE Count} & 0.400 & 0.350 & 0.460 & 0.890 & 0.910 & 0.920 \\
		\hline
		\texttt{TRUH}  & 0.000 & 0.000 & 0.000 & 0.000 & 0.000 & 0.000 \\
		\hline
	\end{tabular}%
	\label{tab:exp3set1}%
\end{table}%
\begin{table}[htbp]
	\centering
	\caption{Rejection rates at $5\%$ level of significance: Experiment 3 and Scenario II wherein $H_0: G \in \mathcal{F}(F_0)$ is false.}
	\begin{tabular}{lcccccc}
		\hline
		& \multicolumn{3}{c}{$m = 500,~n = 10$} & \multicolumn{3}{c}{$m = 2000,~n = 40$} \\
		\hline
		\multicolumn{1}{c}{Method} & $d=5$   & $d=15$  & $d=30$  & $d=5$   & $d=15$  & $d=30$ \\
		\hline
		\texttt{Energy} & 0.850 & 0.920 & 0.940 & 1.000 & 1.000 & 1.000 \\
		\texttt{Crossmatch} & 0.460 & 0.410 & 0.590 & 0.820 & 0.730 & 0.970 \\
		\texttt{E Count} & 0.410 & 0.520 & 0.730 & 0.890 & 0.990 & 1.000 \\
		\texttt{GE Count} & 0.410 & 0.480 & 0.730 & 0.920 & 0.960 & 1.000 \\
		\texttt{WE Count} & 0.550 & 0.580 & 0.780 & 0.920 & 0.920 & 1.000 \\
		\texttt{MTE Count} & 0.590 & 0.570 & 0.810 & 0.900 & 0.960 & 1.000 \\
		\hline
		\texttt{TRUH}  & 0.760 & 0.940 & 0.980 & 0.970 & 1.000 & 1.000 \\
		\hline
	\end{tabular}%
	\label{tab:exp3set2}%
\end{table}%
Tables \ref{tab:exp3set1} and \ref{tab:exp3set2} report the rejection rates for $100$ repetitions of the test when $\bm p$ is held fixed across these repetitions. For scenario I (Table \ref{tab:exp3set1}), we see that \texttt{TRUH}, unlike the other six tests, does not excessively reject the null hypothesis and is the only conservative test. 
In scenario II (Table \ref{tab:exp3set2}) when $n=10$, the \texttt{TRUH} and \texttt{Energy} tests dominate all the remaining tests and reject $H_0$ in more than $80\%$ of the testing instances. However when $n=40$, all tests are competitive, with the exception of the \texttt{Crossmatch} for $d<30$. Overall, across the above two zero-inflated scenarios, \texttt{TRUH} is both conservative and powerful against departures from the null hypothesis $H_0:G\in\mc{F}(F_0)$. 

\section{Remodeling Analysis of HIV-Infected T Cells}
\label{sec:realdata} 
In this section, we analyze the data collected in \citet{cavrois2017mass}. It contains protein expressions of uninfected and HIV infected CD4 (which is a protein found on the surface of immune cells) positive tonsillar T cells.  
We show that existing two-sample tests \red{that rely on a simple null hypothesis for inference, may lead to biologically incorrect inference when testing the composite null hypothesis of equation \eqref{eq:setup-3}.}
\red{Our proposed \texttt{TRUH} hypothesis testing framework, on the other hand, is proficient at detecting $H_0:G\in\mc{F}(F_0)$ and powerful against departures from $H_0$.}

As discussed in section \ref{sec:motivation}, the goal in \citet{cavrois2017mass} was to conduct a mass cytometric assessment of subsets of CD4+ T cells that support HIV entry and viral infection in humans using two variants of the HIV virus: \texttt{Nef rich HIV} and \texttt{Nef deficient HIV}. It is known in the immunology literature, that \texttt{Nef-rich} cells are more prone to viral remodeling \citep{basmaciogullari2014activity}. 
The data set we analyze here contains uninfected and infected data from two different sets of experiments. Both the experiments have four replications based on tonsillar T cells from 4 healthy donors. In Experiment I, the infection was done by \texttt{Nef-rich} HIV where as in Experiment II the infection was done by \texttt{Nef-deficient} HIV. We expect remodeling, if any, in the infected cells to be higher in experiment I than in experiment II compared to their respective baseline uninfected populations. 

The cells in the data were phenotyped in a 38 parameter CyToF \citep{bendall2014single} panel after allowing 4 days for infection. The panel used $3$ markers to classify the cells as uninfected or infected which leaves $d=35$ of the original $38$ markers for our analyses. For donor $r$, let  $\bm{U}_{m, r}=\{U_{1,r},\ldots,U_{m, r}\}$ denote the uninfected sample where each $U_{j, r}$ is a $d$ dimensional vector of arcsin transformed marker expression values with cdf $F_{0}$. We assume that the heterogeneity in the uninfected population is \red{captured} 
by $K$ heterogeneous cellular subgroups each having unimodal probability distribution functions with cdfs $F_1, F_2, \ldots, F_K$  and mixing proportions $w_1, w_2, \ldots, w_K$, such that $F_{0}$ is of the form represented in equation \eqref{eq:setup-1}.
We observe the virus infected sample  $\bm{V}_{n, r}=\{V_{1, r},\ldots, V_{n, r}\}$ consisting  of $n$ i.i.d. $d$-dimensional arcsin transformed observations from $G$ and the goal is to test $H_0: G \in \mathcal{F}(F_0)$ versus $H_A: G \notin \mathcal{F}(F_0)$, where $\mathcal{F}(F_0)$ is the convex hull of  $\{F_1, \ldots, F_K\}$ as defined in equation \eqref{eq:setup-2}. Note that rejection of the null hypothesis would indicate that the distribution of the marker expressions under infection is different from $F_0$ and any convex combination of its components, thus providing evidence in favor of remodeling. 
\red{Virologist study remodeling in virus infected cells in reference to the expressions of Bystander cells. 
	In panels of cells subjected to infection by the virus, not all the cells get infected. Bystanders are those cells which are not directly infected by the virus but are neighbors of virus infected cells. In these experiments, it was seen that when $\tau_{fc}$ is set to 1 then even Bystander cells exhibit remodeling in some experiments. However, when $\tau_{fc}$ is set to $1.1$, there is no remodeling in the Bystander population in any experiments.  Thus, to detect biologically relevant cases of remodeling and avoid discovering benign instances, we use  $\tau_{fc}=1.1$ through out this section to obtain the bootstrapped null distribution of the \texttt{TRUH} statistic.}

Among the 35 markers considered here, it is known that the expressions of the four markers \texttt{CD4, CCR5, CD28,} and \texttt{CD62L} are 
changed due to HIV infection \red{ and these four markers play a significant role in HIV induced remodeling \citep{garcia1991serine,matheson2015cell,michel2005nef,ross1999inhibition,swigut2001mechanism,vassena2015hiv}}. Consider two testing problems: (A) in which we test the hypothesis for all 35 markers, and (B) in which we test the hypothesis of viral remodeling on 31 markers leaving aside the four markers which are known to be remodeled by HIV. Thus, here we have four different cases on which we conduct the tests of viral remodeling, viz.,
\begin{itemize}
	\item CASE 1 corresponds to Experiment I A where we test viral remodeling on \texttt{Nef-rich} infected cells based on all 35 markers including the four which are known to be remodeled. 
	\item CASE 2 corresponds to Experiment I B where we test viral remodeling on \texttt{Nef-rich} infected cells based on 31 markers which are known to be mainly invariant under HIV infection.
	\item CASE 3 corresponds to Experiment II A where we test viral remodeling on \texttt{Nef-deficient} infected cells based on all 35 markers including the four which are known to be remodeled. 
	\item CASE 4 corresponds to Experiment II B where we test viral remodeling on \texttt{Nef-deficient} infected cells based on 31 markers which are known to be mainly invariant under HIV infection.
\end{itemize}
In all the four cases, we have four replications corresponding to four donors. It has been established through validation experiments in \citet{cavrois2017mass} that there is no remodeling but only preferential infection in cases 2 and 4 whereas cases 1 and 3 exhibits remodeling with the intensity of remodeling being much higher in the former than the later. Biologically, it corresponds to the fact that there is \texttt{Nef}-independent remodeling but the intensity of remodeling is higher in presence of \texttt{Nef}. Also, remodeling in cellular expressions is confined to the four markers \texttt{CD4, CCR5, CD28,} and \texttt{CD62L} in the set of markers considered in the study. 
\begin{figure}[!b]
	\centering
	{\includegraphics[width=1\linewidth]{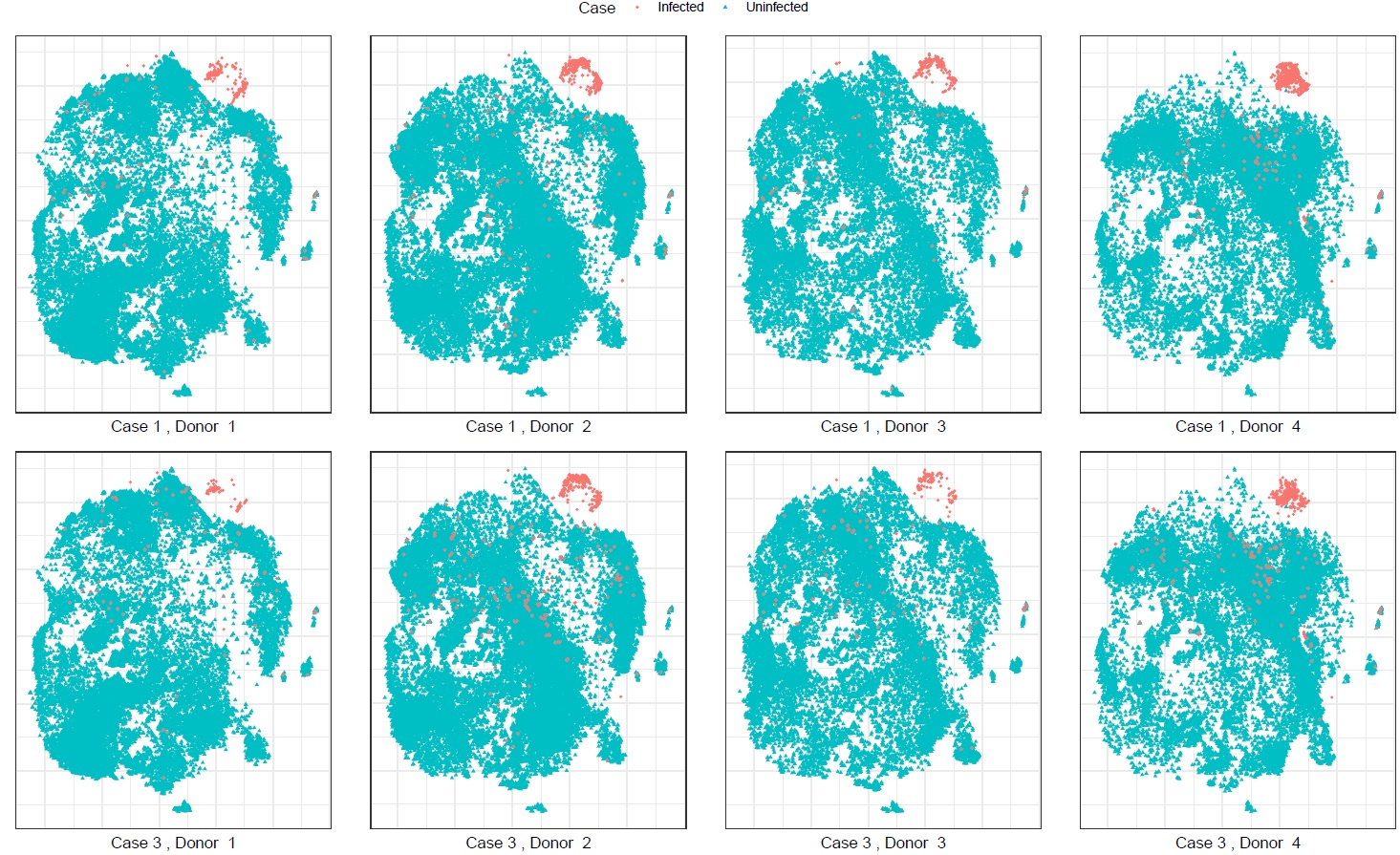}}
	\caption{This is a t-SNE plot \citep{maaten2008visualizing} of the data for Cases 1 and 3 where the $d=35$ dimensional uninfected and infected cellular expression levels are projected to a two dimensional space for each of the four donors.}
	\label{fig:tsne_plot1}
\end{figure}
\begin{figure}[!t]
	\centering
	{\includegraphics[width=1\linewidth]{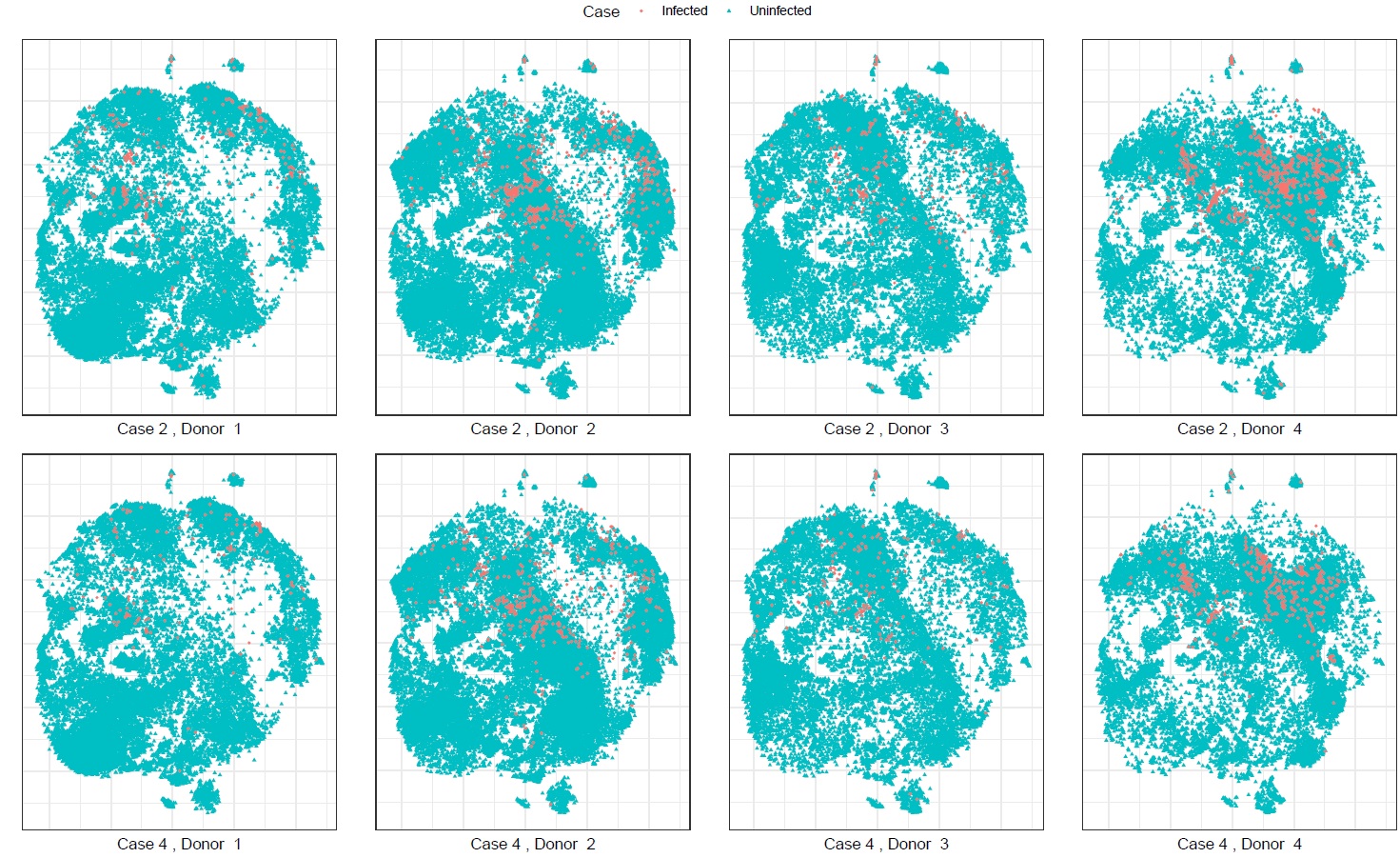}}
	\caption{This is a t-SNE plot of the data for Cases 2 and 4 where the $d=31$ dimensional uninfected and infected cellular expression levels are projected to a two dimensional space for each of the four donors.}
	\label{fig:tsne_plot2}
\end{figure}
\red{Figures \ref{fig:tsne_plot1} and \ref{fig:tsne_plot2}  present t-SNE plots \citep{maaten2008visualizing} of the data where the $d$ dimensional uninfected and infected cellular expression levels are projected to a two dimensional space for each of the four donors across the four cases. While these plots exhibit the underlying heterogeneity in the uninfected sample and the sample size imbalance, instances of remodeling are also visible in cases 1 and 3 (figure \ref{fig:tsne_plot1}) wherein a relatively large fraction of the infected cells in red occupy a distinct position in the two dimensional space with no overlap with their uninfected counterparts.} 

For conducting statistical hypothesis test for the above four cases, along with our proposed \texttt{TRUH} procedure, we also use the siz other competing tests statistics described in section \ref{sec:sims} which are the \texttt{Energy} test \citep{aslan2005new}, \texttt{CrossMatch} \citep{rosenbaum2005exact}, \texttt{E Count}\citep{friedman1979multivariate}, \texttt{GE Count} \citep{chen2017new}, \texttt{WE Count} \citep{chen2018weighted} and \texttt{MTE Count} \citep{zhangmaxtypec}. \red{As discussed in section \ref{sec:sims} these six testing procedures are not designed to test the composite null hypothesis of equation \eqref{eq:setup-3} and rely on a simple null hypothesis $H_0:F_0=G$ for inference. In this section we highlight the biologically incorrect inference that may result when these tests are used for testing the composite null hypothesis of no remodeling.}

Figure \ref{fig:realdata} presents the values of the \texttt{TRUH} statistic and the $2.5\text{-th}, 50\text{-th}, 97.5\text{-th}$ percentiles of the associated null distribution. From the plots, it is evident that at 5\% level our proposed procedure correctly captures the biological phenomena of remodeling or no remodeling across the four cases. The other six tests fail to correctly detect the phenomena in some of the four cases due to heterogeneity in the data. Next, we describe the results in further details. 
\begin{figure}[!b]
	\centering
	\includegraphics[width=0.95\linewidth]{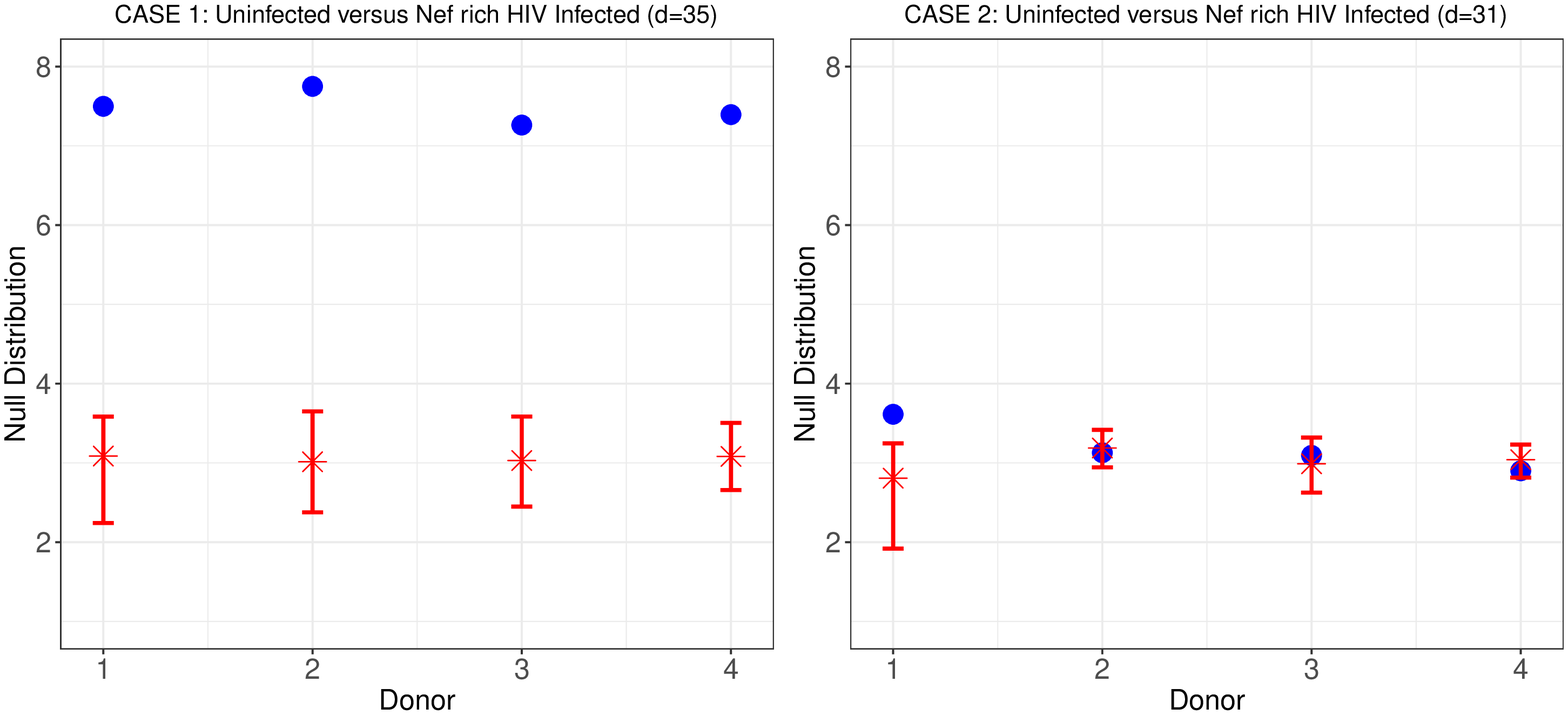}
	\includegraphics[width=0.95\linewidth]{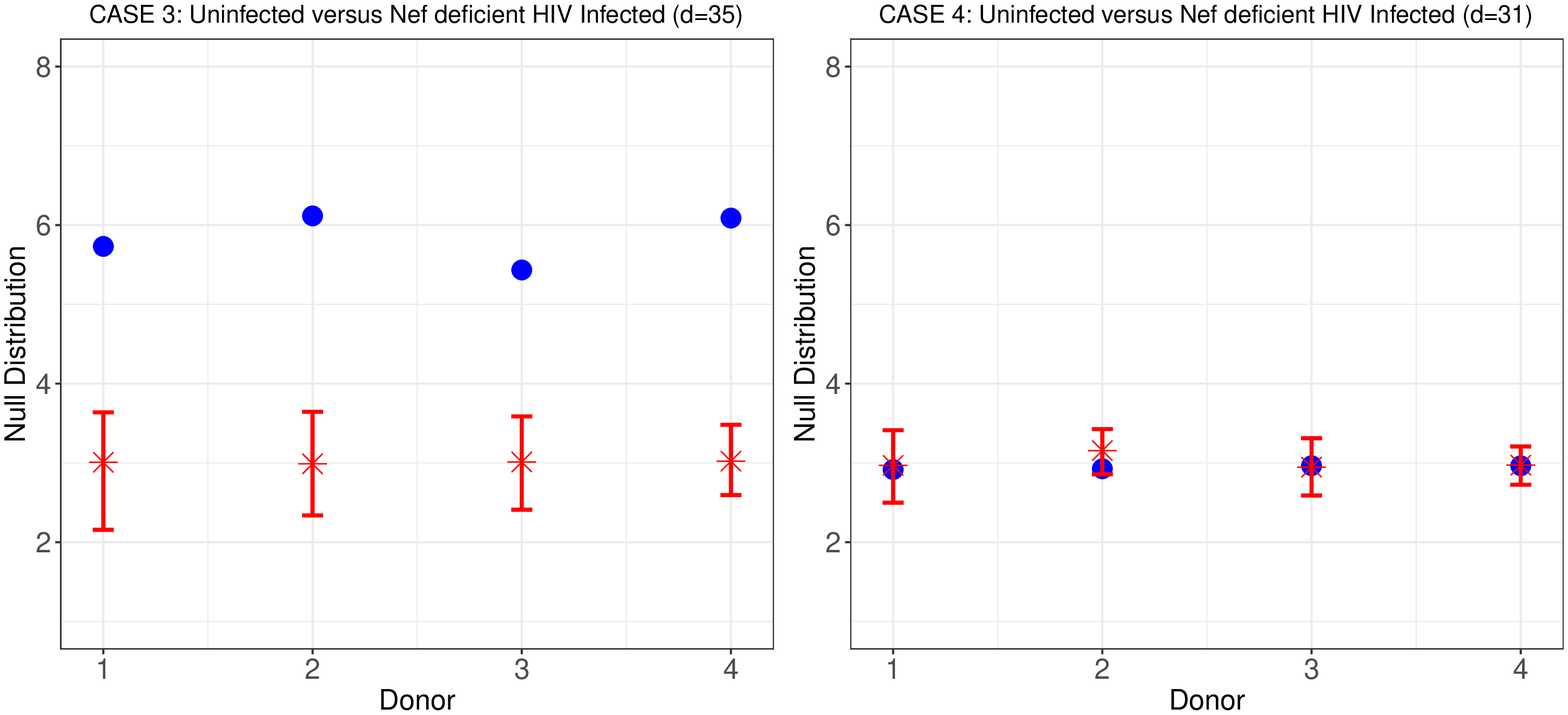}
	\caption{\small{Null distribution of the \texttt{TRUH} statistic under cases 1-4. The blue dots are magnitudes of \texttt{TRUH} statistic for each donor under the four cases while the red bars indicate the $2.5\text{-th}, 50\text{-th}$ and $97.5\text{-th}$ percentiles of the bootstrapped null distribution obtained from algorithm \ref{algo:max} with $\tau_{fc}=1.1$.}}
	\label{fig:realdata}
\end{figure} 
In Tables \ref{tab:test3}  and \ref{tab:test5} we report the p-values of the seven competing tests statistics for testing remodeling under HIV infection in \texttt{Nef-rich} environment. 
In Table \ref{tab:test3} all seven tests reject the null hypothesis of no remodeling, thus verifying that CD4+ T cells exhibit remodeling under the influence of \texttt{Nef} rich HIV infection. In Table \ref{tab:test5}, however, we present the p-values of the tests when the four cell surface markers, \texttt{CD4, CCR5, CD28,} and \texttt{CD62L}, known to be down regulated by \texttt{Nef}, were removed from our analysis $(d=31)$. Other than donor 1, \texttt{TRUH} indicates no remodeling in this scenario for the remaining three donors which is expected given the mechanism of remodeling that \texttt{Nef} pursues by down-regulating \texttt{CD4, CCR5, CD28,} and \texttt{CD62L} \citep{swigut2001mechanism}. The absence of these four cell markers from the uninfected and infected samples reduces the phenotypic gap between these samples as measured through their surface markers. The top row in Figure \ref{fig:realdata} shows that while the null distribution shifts down from CASE 1 (left plot) to CASE 2 (right plot) across all four donors, the drop in the magnitude of the \texttt{TRUH} statistic is far more \red{substantial} 
when the four surface markers are excluded. The remaining six test statistics appear to be insensitive to these subtle changes in the uninfected and infected samples across the two scenarios and, continue to detect remodeling in Case 2 which is actually no remodeling but preferential infection. This demonstrates their inability to handle heterogeneity in the data that \texttt{TRUH} tackles via the composite null testing framework of equations \eqref{eq:setup-1}-\eqref{eq:setup-3}. 

In Tables \ref{tab:test4} and  \ref{tab:test6}, we present the p-values of the seven test statistics for testing the null hypothesis $H_0$ of no remodeling when the HIV infected sample lacks the critical \texttt{Nef} gene (see Construction and validation of reporter viruses in Supplemental Experimental Procedures of \citet{cavrois2017mass} for details around the generation of \texttt{Nef-deficient} HIV infected cells). 
We see that \texttt{TRUH} rejects the null hypothesis of no remodeling in CASE 3 (Table \ref{tab:test4}) while fails to do so in CASE 4 (Table \ref{tab:test6}), thus corroborating the biological phenomena that (a) \texttt{Nef} independent remodeling is prevalent in HIV infected cells and, (b) even in the  absence of \texttt{Nef}, the down regulation of the four surface markers by other mechanisms contributes to remodeling. The bottom row in Figure \ref{fig:realdata} presents the values of the \texttt{TRUH} statistic and the $2.5\text{-th}, 50\text{-th}, 97.5\text{-th}$ percentiles of the associated null distribution. Similar observations from the top row continue to hold for cases 3 and 4 in the bottom row of Figure \ref{fig:realdata} wherein the drop in the magnitude of \texttt{TRUH} statistic is far more significant when the four surface markers are excluded. Moreover, from Figure \ref{fig:realdata}, we see that for every donor the \texttt{TRUH} statistic obeys a rank ordering across the scenarios which is of the form $\texttt{TRUH}_1>\texttt{TRUH}_3>\texttt{TRUH}_2>\texttt{TRUH}_4$ where $\texttt{TRUH}_s$ is the magnitude of the \texttt{TRUH} statistic under cases $s=1,\ldots,4$. This is not accidental for the relative strength of remodeling is known to be highest under the influence of \texttt{Nef-rich} HIV infection and more so when \texttt{Nef} down-regulates the four cell surface markers, \texttt{CD4, CCR5, CD28,} and \texttt{CD62L}. As was seen in cases 1 and 2, the remaining six tests continue to side in favor of remodeling in both cases 3 and 4, thus reflecting their relative lack of conservatism in detecting remodeling under our composite null testing framework.

\begin{table}[!t]
	\centering
	\caption{p-values in CASE 1: Uninfected versus \texttt{Nef-rich} HIV Infected for entire 35 markers.}
	\scalebox{0.8}{
		\begin{tabular}{lcccc}
			\hline
			& \multicolumn{1}{c}{Donor 1} & \multicolumn{1}{c}{Donor 2} & \multicolumn{1}{c}{Donor 3} & \multicolumn{1}{c}{Donor 4} \\
			Tests & \multicolumn{1}{l}{$m = 24,984, n = 245$} & \multicolumn{1}{l}{$m = 31,552, n = 521$} & \multicolumn{1}{l}{$m = 17,704, n = 211$} & \multicolumn{1}{l}{$m = 22,830, n = 660$} \\
			\hline
			\texttt{Energy} &  $<0.001$    & $<0.001$      &  $<0.001$     & $<0.001$ \\
			\texttt{CrossMatch} & $0.005$      &   $0.005$    &  $0.005$     & $0.005$ \\
			\texttt{E Count} &   $<0.001$    & $<0.001$      &  $<0.001$     & $<0.001$ \\
			\texttt{GE Count} &  $<0.001$    & $<0.001$      &  $<0.001$     & $<0.001$ \\
			\texttt{WE Count} &   $<0.001$    & $<0.001$      &  $<0.001$     & $<0.001$  \\
			\texttt{MTE Count} &   $<0.001$    & $<0.001$      &  $<0.001$     & $<0.001$ \\
			\texttt{TRUH}  &   $<0.001$    & $<0.001$      &  $<0.001$     & $<0.001$ \\
			\hline
	\end{tabular}}%
	\label{tab:test3}%
	\vspace{10pt}
	\caption{p-values in CASE 2: Uninfected versus \texttt{Nef-rich} HIV Infected for 31 invariant markers.}
	\scalebox{0.8}{
		\begin{tabular}{lcccc}
			\hline
			& \multicolumn{1}{c}{Donor 1} & \multicolumn{1}{c}{Donor 2} & \multicolumn{1}{c}{Donor 3} & \multicolumn{1}{c}{Donor 4} \\
			Tests & \multicolumn{1}{l}{$m = 24,984, n = 245$} & \multicolumn{1}{l}{$m = 31,552, n = 521$} & \multicolumn{1}{l}{$m = 17,704, n = 211$} & \multicolumn{1}{l}{$m = 22,830, n = 660$} \\
			\hline
			\texttt{Energy} &  $<0.001$    & $<0.001$      &  $<0.001$     & $<0.001$ \\
			\texttt{CrossMatch} & $0.005$      &   $0.005$    &  $0.005$     & $0.005$ \\
			\texttt{E Count} &   $<0.001$    & $<0.001$      &  $<0.001$     & $<0.001$ \\
			\texttt{GE Count} &  $<0.001$    & $<0.001$      &  $<0.001$     & $<0.001$ \\
			\texttt{WE Count} &   $<0.001$    & $<0.001$      &  $<0.001$     & $<0.001$  \\
			\texttt{MTE Count} &   $<0.001$    & $<0.001$      &  $<0.001$     & $<0.001$ \\
			\texttt{TRUH}  &  $<0.001$     &   $0.67$    &  $0.274$     & $0.914$ \\
			\hline
	\end{tabular}}
	\label{tab:test5}%
\end{table}%

\begin{table}[!h]
	\centering
	\caption{p-values in CASE 3: Uninfected versus \texttt{Nef-deficient} HIV Infected for the entire 35 markers.}
	\scalebox{0.8}{
		\begin{tabular}{lcccc}
			\hline
			& \multicolumn{1}{c}{Donor 1} & \multicolumn{1}{c}{Donor 2} & \multicolumn{1}{c}{Donor 3} & \multicolumn{1}{c}{Donor 4} \\
			Tests & \multicolumn{1}{l}{$m = 24,984, n = 129$} & \multicolumn{1}{l}{$m = 31,552, n = 382$} & \multicolumn{1}{l}{$m = 17,704, n = 174$} & \multicolumn{1}{l}{$m = 22,830, n = 440$} \\
			\hline
			\texttt{Energy} &  $<0.001$    & $<0.001$      &  $<0.001$     & $<0.001$ \\
			\texttt{CrossMatch} & $0.005$      &   $0.005$    &  $0.005$     & $0.005$ \\
			\texttt{E Count} &   $<0.001$    & $<0.001$      &  $<0.001$     & $<0.001$ \\
			\texttt{GE Count} &  $<0.001$    & $<0.001$      &  $<0.001$     & $<0.001$ \\
			\texttt{WE Count} &   $<0.001$    & $<0.001$      &  $<0.001$     & $<0.001$  \\
			\texttt{MTE Count} &   $<0.001$    & $<0.001$      &  $<0.001$     & $<0.001$ \\
			\texttt{TRUH}  &  $<0.001$    &  $<0.001$     &  $<0.001$     &  $<0.001$  \\
			\hline
	\end{tabular}}%
	\label{tab:test4}%
	\vspace{10pt}
	\caption{p-values in CASE 4: Uninfected versus \texttt{Nef-deficient} HIV Infected for 31 invariant markers.}
	\scalebox{0.8}{
		\begin{tabular}{lcccc}
			\hline
			& \multicolumn{1}{c}{Donor 1} & \multicolumn{1}{c}{Donor 2} & \multicolumn{1}{c}{Donor 3} & \multicolumn{1}{c}{Donor 4} \\
			Tests & \multicolumn{1}{l}{$m = 24,984, n = 129$} & \multicolumn{1}{l}{$m = 31,552, n = 382$} & \multicolumn{1}{l}{$m = 17,704, n = 174$} & \multicolumn{1}{l}{$m = 22,830, n = 440$} \\
			\hline
			\texttt{Energy} &  $<0.001$    & $<0.001$      &  $<0.001$     & $<0.001$ \\
			\texttt{CrossMatch} & $0.005$      &   $0.005$    &  $0.005$     & $0.005$ \\
			\texttt{E Count} &   $<0.001$    & $<0.001$      &  $<0.001$     & $<0.001$ \\
			\texttt{GE Count} &  $<0.001$    & $<0.001$      &  $<0.001$     & $<0.001$ \\
			\texttt{WE Count} &   $<0.001$    & $<0.001$      &  $<0.001$     & $<0.001$  \\
			\texttt{MTE Count} &   $<0.001$    & $<0.001$      &  $<0.001$     & $<0.001$ \\
			\texttt{TRUH}  &   $0.58$    &  $0.94$     &  $0.464$     & $0.524$ \\
			\hline
	\end{tabular}}
	\label{tab:test6}%
\end{table}%
{The remodeling analysis of the HIV-infected T Cells reveals that our proposed testing procedure, \texttt{TRUH}, conforms to the biologically validated phenomenon of remodeling of human tonsillar T cells under both \texttt{Nef-rich} (case 1) and \texttt{Nef-deficient} (case 3) HIV infection. However unlike traditional tests that continue to infer remodeling in cases 2 and 4, \texttt{TRUH} detects preferential infection and concludes that phenotypic differences between the HIV infected and uninfected T cells are primarily driven by variations in the expression levels of \texttt{CD4, CCR5, CD28,} and \texttt{CD62L} across the uninfected and infected cells. Moreover, through cases 1 and 2, \texttt{TRUH} corroborates the findings in \citet{chaudhuri2007downregulation,michel2005nef,swigut2001mechanism,vassena2015hiv} that HIV remodeling of the T cells is driven by \texttt{Nef} dependent down-regulation of \texttt{CD4, CCR5, CD28, CD62L} while through cases 3 and 4 \texttt{TRUH} reveals \texttt{Nef} independent remodeling of T cells as evidenced in \citet{cavrois2017mass}.}

\section{Optimality Properties of the \texttt{TRUH} Statistic} 
\label{sec:Tnm_limit}

In this section we derive the $L_2$-limit of the proposed test statistic $T_{m, n}$ in the usual limiting regime where the sample sizes  $m, n \to \infty$, such that $n/m \to \rho >0$. This can be used to choose a cut-off and construct a test based on  $T_{m, n}$, and  show asymptotic consistency for biologically relevant location alternatives.  

Recall, that the uninfected and infected samples are denoted as
\begin{align}\label{2}
\bm U_m=\{U_1, \ldots, U_m\} \quad \text{ and } \quad \bm V_n =\{V_1, \ldots, V_n\}, 
\end{align} 
which are i.i.d. samples from two unknown densities $f_0$ and $g$ in $\mbb{R}^d$, respectively. 
To derive the limit of $T_{m, n}$ we need certain integrability/moment assumptions on $f_0$ and $g$.  

\begin{assumption}\label{assumption:TMn} The densities $f_0$ and $g$ have a common support $S \subseteq \mbb{R}^d$ and satisfy either one of the following two assumptions depending on the dimension:  
	\begin{enumerate}
		\item For $d \leq 2$, the support $S$ is compact (with a non-empty interior) and $f_0$ and $g$ are bounded away from zero on $S$. 
		
		\item For $d \geq 3$,  $f\red{_0}$ and $g$  satisfy the following conditions: $\int_S f_0(y)^{1-\frac{1}{d}} \mathrm d y < \infty$, $\int_S f_0(y)^{-\frac{1}{d}} g(y) \mathrm d y < \infty$, and $\int_S |y|^r f_0(y) \mathrm d y < \infty$, $\int_S |y|^r g(y) \mathrm d y < \infty$, for some $r > {d}/{(d-2)}$.  
	\end{enumerate}
\end{assumption}

To describe the limit of $T_{m, n}$ we need a few definitions: For $\lambda >0$, denote by $\mc{P}_\lambda$ the homogeneous Poisson process of intensity \red{$\lambda$} 
in $\mbb{R}^d$, and $\mc{P}_{\lambda}^x = \mc{P}_{\lambda} \cup \{x\}$, for $x\in \mbb{R}^d$. Now, define the following two quantities: 
\begin{align}\label{eq:z12}
\zeta_1(\pmb 0, \mathcal P_1) = \inf_{b \in \mc{P}_1}||b|| \quad \text{and} \quad \zeta_2(\pmb 0, \mathcal P_1) = \inf_{b \in \mc{P}_1\backslash N(\pmb 0, \mc{P}_1)} || N(\pmb 0, \mc{P}_1)-b||,
\end{align}
that is, the distance from the origin $\pmb 0$ in $\mbb{R}^d$ to its nearest neighbor in the Poisson process $\mc{P}_1$, and the distance of this point to its neighbor in $\mc{P}_1$, respectively. 

\begin{theorem}\label{THM:EXPTNM} Let $T_{m, n}$ be as in \eqref{eq:TNM}. Then, for $f_0$ and $g$ as in Assumption $\ref{assumption:TMn}$ above, as $m, n \rightarrow \infty$ such that $n/m \rightarrow \rho$, 
	\begin{align}\label{eq:TNM_limit}
	T_{m, n}  \stackrel{L_2} \rightarrow  \varphi(f_0, g, \rho)= \rho^{\frac{1}{d}} \Delta_d \int \frac{g(y)}{f_0(y)^{\frac{1}{d}} } \mathrm d y,
	\end{align}
	with $\Delta_d= (\zeta_2-\zeta_1 )$, where
	\begin{itemize}
		\item[--] $\zeta_1= \mbb{E} \zeta_1(\pmb 0, \mathcal P_1)$, the expected distance from the origin in $\bm 0 \in \mbb{R}^d$ to its nearest neighbor in $\mc{P}_1$, and 
		\item[--]   $\zeta_2= \mbb{E} \zeta_2(\pmb 0, \mathcal P_1)$, the excepted distance between the nearest neighbor of the origin in $\mc{P}_1$ to its nearest neighbor in $\mc{P}_1$. 
	\end{itemize}
\end{theorem}

The above theorem gives the $L_2$-limit of the test statistic for general distributions $f_0$ and $g$. The proof of the theorem, which is given in the supplementary materials (Section B), uses the machinery of  geometric stabilization, introduced by \citet{py}, which obtains the asymptotics of nearest neighbor based functionals in terms of functionals defined on a homogeneous Poisson process. Before we discuss how the result in Theorem \ref{THM:EXPTNM} can be used to   construct a test based on $T_{m, n}$ for the hypothesis \eqref{eq:setup-2}, we discuss some properties and the consequences of the limit in \eqref{eq:TNM_limit}: 

\begin{itemize}
	
	\item Note that the finiteness of the limit in \eqref{eq:TNM_limit} is ensured by Assumption \ref{assumption:TMn}.  For $d \geq 3$, the moment conditions in Assumption \ref{assumption:TMn} are required to establish the $L_2$ convergence in \eqref{eq:TNM_limit}.  This assumption can be relaxed to $\int_S |y|^r f_0(y) \mathrm d y < \infty$ and $\int_S |y|^r g(y) \mathrm d y < \infty$, for some $r > {d}/{(d-1)}$, if we are only interested in $L_1$ convergence (by combining the proof of Theorem \ref{THM:EXPTNM} with that of \cite[Proposition 3.2]{py}). However, this still does not apply for $d=1$, where it is necessary to assume the compactness of the support, in order to ensure that the limit in \eqref{eq:TNM_limit} is finite.  This is a well-known constraint which arises in a large family of random geometric graphs, while dealing with the asymptotics of edge-lengths (see, for example, \cite[Theorem 1.1]{py} and the references therein). Even though the compactness assumption technically rules out some natural distributions, from a practical standpoint, there is no real concern because one can approximate the univariate density by truncating it to a large interval, on which the above result applies. \red{Incidentally, there has been recent work on relaxing the compactness and density bounded below  assumptions in the related problems of nearest-neighbor classification \cite{cannings2017local,gadat2016classification} and entropy estimation \cite{berrett2019efficient}, which could provide useful insights on how to relax these assumptions from Theorem \ref{THM:EXPTNM}, and what are the effects of tail behavior on the heterogeneity testing problem.}

	\item  Note that $\zeta_1$ and $\zeta_2$ are both constants, which depend only on the dimension $d$. In fact, $\zeta_1$ has a closed form expression which can be easily derived. To this end, denote by $V_d$ and $S_d$ the volume and the surface area of the unit ball in $\mbb{R}^d$, respectively. It is easy to verify that $S_d=d V_d$. Moreover, for $r > 0$ and $x \in \mbb{R}^d$, denote by $B(x, r)$ the ball of radius $r$ centered at $x \in \mbb{R}^d$.  Then, using the observation that a point $b$ is the nearest neighbor of the origin, if are there no points of the Poisson process $\mc{P}_1$ in the ball $B(0, ||b||)$, it follows that 
	\begin{align*}
	\zeta_1=\mbb{E} (\zeta_1(\pmb 0, \mathcal P_1)   = \int ||b|| \P(b=N(\bm 0, \mc{P}_1^{0, b})) \mathrm db  
	= S_d \int_{0}^\infty t^d e^{-V_d t^d } \mathrm d t~, 
	\end{align*} 
	which, by the change of variable $x=V_dt^d$ equals
	\begin{align}   \label{eq:z1_limit}
	\left(\frac{1}{V_d}\right)^{\frac{1}{d}} \int_0^\infty x^{\frac{1}{d}} e^{-x} \mathrm dx  
	= \left(\frac{1}{V_d}\right)^{\frac{1}{d}}\Gamma\left(\frac{d+1}{d}\right),
	\end{align} 
	where $\Gamma(\cdot)$ denotes the Gamma function. 

\end{itemize}


Theorem \ref{THM:EXPTNM} shows that for $K$ fixed densities $f_1, \ldots, f_K$, and $f_0=\sum_{a=1}^K w_a f_a$, 
\begin{align}\label{eq:fg}
\sup_{g \in \mc{F}(f_0)} \varphi(f_0, g, \rho)&=\rho^{\frac{1}{d}} \Delta_d \sup_{\lambda_1, \lambda_2, \ldots, \lambda_K}\sum_{a=1}^K \lambda_a \int \frac{f_a(y)}{\left(\sum_{b=1}^K w_b f_b(y)\right)^{\frac{1}{d}}} \mathrm d y \nonumber \\ 
&=\rho^{\frac{1}{d}} \Delta_d \max_{1 \leq a \leq K} \left\{ \int \frac{\lambda_a  f_a(y)}{\left(\sum_{b=1}^K w_b f_b(y) \right)^{\frac{1}{d}}} \mathrm d y \right\},
\end{align}
where the last step uses the fact that $\lambda_a \in [0, 1]$, for $1 \leq a \leq K$, and $\sum_{a=1}^K \lambda_K=1$. Note that the RHS above is unknown, because the densities  $f_1, \ldots, f_K$, the weights  $w_1, \ldots, w_K$, as well as the number $K$ of mixture components, are all unknown. However, if we can consistently estimate the RHS of \eqref{eq:fg}, then the test which rejects $H_0$ in \eqref{eq:setup-3} when $T_{m, n}$ is greater than the estimated value of \eqref{eq:fg}, would have zero asymptotic Type I error and would be powerful whenever $g$ has some separation from the set $\mc{F}(f_0)$ (recall definition in \eqref{eq:setup-2}).

The approach described above is, in general, infeasible because nonparametric estimation of mixture parameters in multivariate problems, especially when the number $K$ is unknown, can often be difficult. In the following, we show how in location families, one can obtain a slightly weaker upper bound on $\varphi(f_0, g, \rho)$, which is free of the unknown parameters, that can  be used to construct a valid and powerful test for the remodeling hypothesis \eqref{eq:setup-3}. To this end, consider $\{p(y|\theta)=p(y-\theta):  \theta \in \Theta\}$ a family of densities indexed by the parameter space $\Theta \subseteq \mbb{R}^d$, where $p: \mbb{R}^d \rightarrow \mbb{R}_{\geq 0}$ such that $\int_{\mbb{R}^d} p(y) \mathrm dy=1$.  Throughout we assume that the densities in the family satisfy Assumption \ref{assumption:TMn}. Suppose the baseline samples $U_1, U_2, \ldots, U_m$ are i.i.d. from the density $f_0(\cdot)=\sum_{a=1}^K w_a \, p(\cdot|\theta_a)$, where  
$\theta_1, \ldots, \theta_K \in \Theta$ are fixed (but unknown), and there exists a known constant $L >0$ such that $w_a \geq L$, for all $1 \leq a \leq K$. 
If the infected samples $V_1, V_2, \ldots, V_n$ are i.i.d. from a density $g$ in $\mbb{R}^d$, then the hypothesis of remodeling \eqref{eq:setup-2}, in this parametric setting, becomes, 
\begin{align}\label{eq:thetaH0H1}
H_0: g \in \mathcal{F}(\bm \theta) \quad \text{versus} \quad H_A: g \notin \mathcal{F}(\bm \theta),
\end{align}
where $\bm \theta=(\theta_1, \ldots, \theta_K)$ and $\mathcal{F}(\bm \theta)$ is defined as follows: 
$$\mathcal{F}(\bm \theta) =\left\{ q(\cdot)= \sum_{a=1}^K \lambda_a p(\cdot|\theta_a) :  \lambda_a \in [0,1], \text{ for } 1\leq a \leq K, \text{ and } \sum_{a=1}^K \lambda_a=1 \right\},$$
is the collection of $K$-mixtures of $p(\cdot|\theta_1), p(\cdot|\theta_2), \ldots, p(\cdot|\theta_K)$. Note that under the null $H_0$, $g(\cdot) =\sum_{a=1}^K \lambda_a  p(\cdot|\theta_a)$, for some $\lambda_1, \lambda_2, \ldots, \lambda_K \in [0,1]$, such that $\sum_{a=1}^K \lambda_a=1$. Then using $\sum_{a=1}^K w_a p(y|\theta_a) > w_b p(y|\theta_b) \geq L p(y|\theta_b)$, for all $b \in \{1,2,\ldots,K\}$, 
\begin{align}
\varphi(f_0, g, \rho) & = \rho^{\frac{1}{d}} \Delta_d \sum_{a=1}^K \lambda_a \int \frac{p(y|\theta_a)}{\left(\sum_{b=1}^K w_b p(y|\theta_b)\right)^{\frac{1}{d}}} \mathrm d y \nonumber \\ 
\label{eq:location_H0} & = \frac{\rho^{\frac{1}{d}} \Delta_d}{L^{\frac{1}{d}}} \int p(z)^{1-\frac{1}{d}} \mathrm d z=\gamma,
\end{align} 
where the last step follows by the change of variable $z=y-\theta_a$. Note that the constant $\gamma$ depends on $L$ (the lower bound on the mixing weights of the baseline population), the dimension $d$, and the base function $p$ defining the location family (which is assumed to be known); but not on the unknown means ($\theta_1, \theta_2, \ldots, \theta_K$), the unknown weights ($w_1, w_2, \ldots, w_K$), or the number of components, and hence can be directly calculated.  \red{This implies that the test which rejects when $T_{m, n} > \gamma$, would have zero asymptotic Type I error, and would also be powerful whenever $g$ has some separation from the set of possible null distributions $\cF(f_0)$, as explained below.}

The corollary below shows how the bound in \eqref{eq:location_H0} can be used to construct a test based on $T_{m, n}$ which is powerful for \red{mixtures of radially symmetric distributions,} 
such as Gaussian mixtures and $t$-mixtures, among others. Hereafter, we assume $p(y)=r(||y||)$ is radially symmetric, where $r: \mbb{R}_{\geq 0} \rightarrow  \mbb{R}_{\geq 0}$ is a uniformly continuous function, such that $\int_{\mbb{R}^d} r(||y||) \mathrm d y = 1$. (Recall, $||y||$ denotes the Euclidean norm of $y \in \R^d$.)

\begin{corollary}\label{COR:H0H1}
	For the testing problem \eqref{eq:thetaH0H1} in the family $\{p(y|\theta) =r(||y-\theta||): \theta \in \Theta\}$, the following hold: 
	\begin{itemize}
		
		\item  For any $g \in \mc{F}(\bm \theta)$, with $\gamma$ as defined in \eqref{eq:location_H0}, we have
		\begin{align}\label{eq:H0}
		\lim_{m, n \rightarrow \infty}\P_{f_0, g}(T_{m, n} > \gamma) = 0~.
		\end{align}
		
		\item There exists $\varepsilon(\gamma)>0$ such that 
		\begin{align}\label{eq:H1}
		\lim_{m, n \rightarrow \infty}\P_{f_0, g}(T_{m, n} > \gamma) =1,
		\end{align}
		for any $g(y) = \sum_{a=1}^K \bar \lambda_a p(y|\theta_a')$ with $\min_{1 \leq a, b \leq K}||\theta_a'-\theta_b|| \, \bm 1\{\bar \lambda_a >0\} \geq \varepsilon(\gamma)$. 
	\end{itemize}
\end{corollary}

The proof of the corollary is given in the supplementary materials (Section C). Note that the condition on $g(y)$ in \eqref{eq:H1} quantifies a natural notion of separation between $g$ and the set $\mc{F}(\bm \theta)$, by assuming that at least one of the mixture means of $g$ is $\varepsilon$-far (in $L_2$-distance) from all the unknown null means of the baseline density. 
Explicit bounds on the separation $\varepsilon(\gamma)$ can be obtained from the proof of Corollary \ref{COR:H0H1}, based on the tail decay of the base density $p$ (details given in supplementary materials, Section C). 



\section{Discussion}
\label{sec:discuss}
We propose a novel nearest neighbor based two-sample test for detecting  changes between the baseline and the case samples, in the presence of heterogeneity, as is often the case in single-cell virology. \red{For integrative analysis involving datasets collected from differerent experiments with varying external conditions, batch-effect corrections are needed before applying our methodology. }
Our testing procedure is specially designed for mass cytometry based techniques  \citep{Bendall11,giesen2014highly} which produces moderate dimensional ($d \sim 50$) cellular characteristics. In the future, it will be interesting to extend our methodology for dealing with single-cell RNA-seq based techniques \citep{huang2018saver,hwang2018single,jaitin2014massively,schiffman2017sideseq}, which can produce highly multivariate phenotypes ($d\sim10^4$). A possible approach can be \red{based on random projections of the $d$ dimensional cellular characteristics to a lower dimensional space and then using our testing procedure on the reduced data.}
Also, it will be interesting to develop efficient testing procedures where the underlying population contains heterogeneous subpopulations with highly varying sizes including some very rare subpopulations. \red{Finally, extending our hypothesis testing framework to distinguish between depletion and enrichment in remodeled cells will be important.}

\section*{Acknowledgements}
{We are grateful to  Ann Arvin, Nadia Roan,  Adrish Sen, Nandini Sen and Nancy Zhang for numerous stimulating discussions. We thank the Editor, the Associate Editor and three referees for constructive suggestions that greatly improved the paper.}

\appendix


\section{Proof of Proposition~1}
\label{sec:pfinconsistency}
Recall that for $U_1,\ldots,U_m$ are i.i.d. $f_0$ and $V_1,\ldots,V_n$ are i.i.d. $g$. Then, in the usual asymptotic regime, by Theorem 2 of  \citet{henze1999multivariate}, almost surely,  
\begin{align}\label{eq:FG_I}
\frac{\mc{R}(\bm U_m, \bm V_n)}{m+n} \stackrel{a.s.}\rightarrow 1- \delta(f_0, g,\rho)
\end{align}
where $\delta(f_0, g,\rho)=\int \frac{f^2(x) + \rho^2  g^2(x)}{(1+\rho)(f_0(x) + \rho g(x))}\mathrm dx$. 

Now, by Remark 1 of  \citet{henze1999multivariate} for any fixed $g \in \mathcal{F}(f_0) \setminus \{f_0\}$, 
$$1-\delta(f_0, g,\rho) < 1- \delta(f_0, f_0,\rho) = \frac{2 \rho}{(1+\rho)^2}.$$  
Note that for any fixed $\alpha \in (0,1/2)$, $\frac{C_{m,n}(\alpha)}{m+n} \to \frac{2 \rho} {(1+\rho)2}$ almost surely. Therefore, by \eqref{eq:FG_I}, for any fixed $g \in \mathcal{F}(f_0) \setminus \{f_0\}$, $\mc{R}(\bm U_m, \bm V_n) < C_{m,n}(\alpha)$ almost surely, and the result follows.

\section{Proof of Theorem 1}
\label{sec:pfthm_CD}

The proof of Theorem 1 is an immediate consequence of the following two lemmas. The first lemma computes the limit of $n^{\frac{1}{d}} \bar D_{m, n}$. 

\begin{lemma}\label{lm:D} Let $D_1, D_2, \ldots, D_n$ be as defined in equation (2.5). Then, under Assumption 1, as $m, n \rightarrow \infty$, 
	\begin{align}\label{eq:lm_D}
	\frac{1}{n^{1-\frac{1}{d}}} \sum_{i=1}^n D_i  \stackrel{L_2} \rightarrow  \rho^{\frac{1}{d}} \zeta_1 \int \frac{g(y)}{f_0(y)^{\frac{1}{d}}} \mathrm d y,
	\end{align}
	where $\zeta_1$ is as defined in the statement of Theorem 1. 
\end{lemma}

The next lemma computes the limit of $n^{\frac{1}{d}} \bar C_{m, n}$, which combined with Lemma \ref{lm:D} completes the proof of Theorem 1. 

\begin{lemma}\label{lm:C} Let $C_1, C_2, \ldots, C_n$ be as defined in equation (2.6). Then, under Assumption 1, as $m, n \rightarrow \infty$, 
	\begin{align}\label{eq:lm_C}
	\frac{1}{n^{1-\frac{1}{d}}}  \sum_{i=1}^n C_i  \stackrel{L_2} \rightarrow  \rho^{\frac{1}{d}} \zeta_2  \int \frac{g(y)}{f_0(y)^{\frac{1}{d}}} \mathrm d y,
	\end{align}
	where $\zeta_2$ is as defined in the statement of Theorem 1. 
\end{lemma}

The proofs of Lemma \ref{lm:D} and Lemma \ref{lm:C} are given below in Section \ref{sec:pf_D} and Section \ref{sec:pf_C}, respectively. We begin with some preliminaries about Poisson processes and stabilization of geometric functionals, introduced by \cite{py}, in Section \ref{sec:preliminaries} below. 

\subsection{Preliminaries}
\label{sec:preliminaries}

%
%

Given $z \in \mbb{R}^d$, denote by $\varphi(z, \mc{Z})$ a measurable $\mbb{R}^+$ valued function defined for all locally finite set $\mc{Z}\subset \mbb{R}^d$ and $z\in \mc{Z}$. If $z\notin \mc{Z}$, then  $\varphi(z, \mc{Z}):=\varphi(z, \mc{Z}\cup\{z\})$. The function $\varphi$ is said to be {\it translation invariant} if $\varphi(y+z, y+\mc{Z})=\varphi(z, \mc{Z})$. \cite{py} defined stabilizing functions as follows:

\begin{defn}(\cite{py}) \label{defn:stabilize}
	For any locally finite point set $\mc{Z}\subset \mbb{R}^d$ and any integer $M \in \mbb{N}$, 
	$$\overline \varphi(\mc{Z}, M):=\sup_{N\in \mbb{N}}\left(\mathrm{esssup}_{\substack{\mc{A}\subset \mbb{R}^d\setminus B(0, M)\\|\mc{A}|=N}}\left\{\varphi(0, \mc{Z}\cap B(0, M) \cup \mc{A})\right\}\right)$$
	and
	$$\underline \varphi(\mc{Z}, M):=\inf_{N\in \mbb{N}}\left(\mathrm{essinf}_{\substack{\mc{A}\subset \mbb{R}^d\setminus B(0, M)\\|\mc{A}|=N}}\left\{\varphi(0, \mc{Z}\cap B(0, M) \cup \mc{A}) \right\}\right),$$
	where the essential supremum/infimum is taken with respect to the Lebesgue measure on $\mbb{R}^{dN}$. The functional $\varphi$ is said to {\it stabilize} $\mc{Z}$ if
	\begin{equation}\label{eq:stabilize}
	\liminf_{M\rightarrow \infty}\underline \varphi(\mc{Z}, M)=\limsup_{M\rightarrow \infty}\overline\varphi(\mc{Z}, M)=\varphi (0, \mc{Z}).
	\end{equation}
\end{defn}

We will be interested in functionals that stabilize almost surely on $\mc{P}_\lambda$, the homogeneous Poisson process with rate $\lambda$ in $\mbb{R}^d$. Note that with probability 1, $\overline \varphi(\mc{P}_\lambda, M)$ is nonincreasing in $M$ and  $\underline \varphi(\mc{P}_\lambda, M)$ is nondecreasing in $M$, therefore, they both converge. The definition of stabilization in \eqref{eq:stabilize} means they converge to the same limit almost surely.  Note that any functional $\varphi(z, \mc{Z})$ which depends only on the points of $\mc{Z}$ within a fixed distance of $z$ is stabilizing on $\mc{P}_\lambda$. In our proofs, we will consider the following two functionals: 

\begin{itemize}
	
	\item For $y \in \mbb{R}^d$, and $\mc{Z} \subset \mbb{R}^d$ finite, define 
	\begin{align}\label{eq:z1}
	\zeta_1(y, \mc{Z}):=\sum_{z \in \mc{Z}}||y-z||  \bm 1\{z=N(y, \mc{Z}) \},\
	\end{align}
	which is the distance from $y$ to its nearest neighbor in $\mc{Z}$. 
	
	\item For $y \in \mbb{R}^d$, and $\mc{Z} \subset \mbb{R}^d$ finite, define 
	\begin{align}\label{eq:z2}
	\zeta_2(y, \mc{Z}):=\sum_{z_1 \in \mc{Z}} \sum_{z_2 \in \mc{Z}\backslash\{z_1\}} ||z_1-z_2||  \bm 1\{z_1=N(y, \mc{Z}) \text{ and } z_2=N(z_1, \mc{Z}\backslash\{y\}) \}, 
	\end{align}
	which is the distance between the nearest neighbor of $y$ in $\mc{Z}$ and its nearest neighbor in $\mc{Z}$. 
	
\end{itemize}
It is easy to verify that both the functionals $\zeta_1(\cdot, \cdot)$ and $\zeta_2(\cdot, \cdot)$ stabilize $\mc{P}_\lambda$, for all $\lambda > 0$. This is because the set of edges incident to
the origin in the directed 1-nearest neighbor (NN) graph\footnote{
	Given a finite set $S \subset \mbb{R}^d$, the directed $1$-nearest neighbor  graph ($1$-NN)  is a graph with vertex set $S$ with a directed edge $(a, b)$, for $a, b \in S$, if $b$ is the nearest neighbor of $a$ in $S$. 
} is unaffected by the addition or removal of points outside a ball of almost surely finite radius \cite[Theorem 2.4]{py}. 


\subsection{Proof of Lemma \ref{lm:D}}
\label{sec:pf_D}

We now proceed to prove Lemma \ref{lm:D}. We begin by noting that $\mbb{E}(\bar D_{m, n})=\mbb{E}(D_1)$ and 
\begin{align}\label{eq:exp_D_I}
\mbb{E}(D_1)=\sum_{j=1}^m \mbb{E} ||V_1-U_j||  \bm 1\{U_j= N(V_1, \bm U_m) \}  & = \mbb{E} \zeta_1(V_1, \bm U_m),
\end{align}
where $\zeta_1(\cdot, \cdot)$ is as defined above in \eqref{eq:z1} and $\bm U_m=\{U_1, U_2, \ldots, U_m\}$ are i.i.d. points from the density $f_0$. Note that, by translation invariance, 
\begin{align}\label{eq:scaled_m}
\zeta_1(\pmb 0, m^{\frac{1}{d}}(\bm U_m-V_1))&:=m^{\frac{1}{d}}\sum_{j=1}^m ||V_1-U_j||  \bm 1\{m^{\frac{1}{d}}(U_j-V_1)= N(\pmb 0, m^{\frac{1}{d}}(\bm U_m-V_1)) \} \nonumber \\ 
&= m^{\frac{1}{d}}\sum_{j=1}^m ||V_1-U_j|| \bm 1\{U_j= N(V_1, \bm U_m) \} \nonumber \\  
& = m^{\frac{1}{d}}  \zeta_1(V_1, \bm U_m).
\end{align}

The following lemma shows that the second moment of $\zeta_1(\pmb 0, m^{\frac{1}{d}}(\bm U_m-V_1))$ is bounded, under  Assumption 1. 

\begin{lemma}\label{lm:moment2} For densities $f_0$ and $g$ as in Assumption 1,
	$$\sup_{m \in \mbb{N}}\mbb{E}\zeta_1(\pmb 0, m^{\frac{1}{d}}(\bm U_m-V_1))^2 \lesssim_d 1.$$  
\end{lemma}

\begin{proof} Note for $d \leq 2$, the result holds trivially, by the boundedness of the support.  Hence, assuming, $d \geq 3$, and taking squares in \eqref{eq:scaled_m} gives, 
	\begin{align}\label{eq:moment2}
	\zeta_1(\pmb 0, & ~m^{\frac{1}{d}}(\bm U_m-V_1))^2 \nonumber \\ 
	&:=m^{\frac{2}{d}}\sum_{1 \leq j_1, j_2 \leq m} ||V_1-U_{j_1}|| \cdot ||V_1-U_{j_2}||  \bm 1\{U_{j_1}= N(V_1, \bm U_m), U_{j_2}= N(V_1, \bm U_m) \} \nonumber \\ 
	& \lesssim m^{\frac{2}{d}}\sum_{j=1}^m ||V_1-U_{j}||^2  \bm 1\{U_{j}= N(V_1, \bm U_m) \},  
	\end{align}
	using the inequality $ab \leq \frac{a^2+b^2}{2}$ and the fact $\sum_{j=1}^m \bm 1\{U_{j}= N(V_1, \bm U_m) \}=1$. Now, for $n$ large enough,  
	\begin{align}\label{eq:phi_m}
	m^{\frac{2}{d}}\mbb{E} \sum_{j=1}^m ||V_1-U_{j}||^2  \bm 1\{U_{j}= N(V_1, \bm U_m) \} & = \frac{m^{\frac{2}{d}}}{n} \mbb{E} \sum_{i=1}^n \sum_{j=1}^m ||V_i-U_{j}||^2  \bm 1\{U_{j}= N(V_i, \bm U_m) \} \nonumber \\  
	& \lesssim_d \frac{1}{n^{1-\frac{2}{d}}} \mbb{E} \phi(\bm V_n, \bm U_m), 
	\end{align}
	where the functional $\phi(A, B):=\sum_{a\in A} \sum_{b \in B} ||a-b||^2 \bm 1\{b=N(a, B)\}$, where $A, B \subset \mbb{R}^d$ are finite and disjoint. Note that for any partition $\{S_0, S_1, \ldots, \}$ of $\mbb{R}^d$, 
	\begin{align}\label{eq:phiab}
	\phi(A, B) \leq \sum_{K=0}^\infty \phi(A \cap S_K, B\cap S_K),
	\end{align}
	that is, the functional $\phi$ is subadditive. (Note that the sum above is, in fact, finite because the sets $A$ and $B$ are finite.) Then by a modification of \cite[Lemma 3.3]{monograph_yukich}, one can obtain the growth bound $\phi(A, B) \leq \mathrm{diam}(A \cup B)^2 |A\cup B|^{\frac{d-2}{d}}$. 
	Now, choosing $S_0$ to be the ball of radius 2 centered at the origin, and $S_K$ to be the annulus centered at the origin with inner radius $2^K$ and outer radius $2^{K+1}$, for $K \geq 1$, it follows from \eqref{eq:phiab} that 
	$$\phi(\bm V_n, \bm U_m) \leq \sum_{K=0}^\infty 2^{2K} \left|\sum_{i=1}^n\bm 1\{V_i \in S_K\} + \sum_{j=1}^m\bm 1\{U_i \in S_K\} \right|^{\frac{d-2}{d}}.$$
	Now, taking expectations above and the Jensen's inequality gives, for $n$ large enough, 
	$$\frac{1}{n^{1-\frac{2}{d}}} \mbb{E} \phi(\bm V_n, \bm U_m) \lesssim_d  \sum_{K=0}^\infty 2^{2K}\P(V_1\in S_K)^{\frac{d-2}{d}} + \sum_{K=0}^\infty 2^{2K}\P(U_1\in S_K)^{\frac{d-2}{d}},$$ 
	both of which are finite by the integrality assumptions on $f_0$ and $g$ (using arguments in \cite[Page 85]{monograph_yukich}). The result now follows by combining the bound above with \eqref{eq:moment2} and \eqref{eq:phi_m}. 
\end{proof}

The lemma above shows that the sequence $\{\zeta_1(\pmb 0, m^{\frac{1}{d}}(\bm U_m-V_1))\}_{M \geq 1}$ is uniformly integrable.  Now, since the functional $\zeta_1(\cdot, \cdot)$ stabilizes on homogeneous Poisson processes, by arguments similar to the proof of \cite[Lemma 8.1]{yukichclt}, it follows that 
\begin{align}\label{eq:Mlimit_I}
\lim_{M \rightarrow \infty}\mbb{E} \zeta_1(\pmb 0, m^{\frac{1}{d}}(\bm U_m-V_1)) = \mbb{E} \zeta_1(\pmb 0, \mathcal P_{f_0(V)})),
\end{align} 
where $\zeta_1(\pmb 0, \mathcal P_1)$ is as defined in equation (5.2), $V$ is a random variable distributed according to the density $g$, and $\mathcal P_{f_0(V)}$ is a Cox process with intensity measure $f_0(V)$, which is a Poison process with a random intensity measure $f_0(V)$.  Conditioning on $V$ gives, 
$$\mbb{E} \zeta_1(\pmb 0, \mathcal P_{f_0(V)})) =\int \mbb{E} \zeta_1(\pmb 0, \mathcal P_{f_0(y)})) g(y) \mathrm dy= \mbb{E} \zeta_1(\pmb 0, \mathcal P_1)) \int \frac{g(y)}{f_0(y)^{\frac{1}{d}}} \mathrm dy,$$
where the last step uses $\mc{P}_\lambda \stackrel{D}=\lambda^{-\frac{1}{d}}  \mc{P}_1$, for any $\lambda >0$.  This implies, by \eqref{eq:exp_D_I}, \eqref{eq:scaled_m}, and \eqref{eq:Mlimit_I}, that  
$$m^{\frac{1}{d}} \mbb{E}(\bar D_{m, n}) = \mbb{E} \zeta_1(\pmb 0, m^{\frac{1}{d}}(\bm U_m-V_1)) \rightarrow \mbb{E} (\zeta_1(\pmb 0, \mathcal P_1)) \int  \frac{g(y)}{f_0(y)^{\frac{1}{d}}} \mathrm d y. $$
Then, recalling $n/m \rightarrow \rho$ gives, 
\begin{align}\label{eq:expDn}
\mbb{E}\left(\frac{1}{n^{1-\frac{1}{d}}} \sum_{i=1}^n D_i \right)   \rightarrow  \rho^{\frac{1}{d}} \mbb{E} (\zeta_1(\pmb 0, \mathcal P_1))  \int \frac{g(y)}{f_0(y)^{\frac{1}{d}}} \mathrm d y,
\end{align}
which establishes the limit in \eqref{eq:lm_D} in expectation.

To complete the proof of the lemma we need to show that the variance of the LHS in \eqref{eq:lm_D} goes to zero. To this end, note that 
\begin{align}\label{eq:variance_D}
\mbb{E}\left(\frac{1}{n^{1-\frac{1}{d}}} \sum_{i=1}^n D_i \right)^2 &= \frac{1}{n^{1-\frac{2}{d}}} \mbb{E} D_1^2  + \frac{n(n-1)}{n^2} n^{\frac{2}{d}} \mbb{E} D_1 D_2 = (1+o(1)) n^{\frac{2}{d}} \mbb{E} D_1 D_2 +o(1),
\end{align}
since $n^{\frac{2}{d}}\mbb{E} D_1^2 \lesssim _d 1$, by   Lemma \ref{lm:moment2}. Next, note that 
\begin{align}
m^{\frac{2}{d}} \mbb{E}(D_1D_2)&=m^{\frac{2}{d}} \sum_{j_1=1}^m \sum_{j_2=1}^m \mbb{E} ||V_1-U_{j_1}|| ||V_2-U_{j_2}||   \bm 1\{U_{j_1}= N(V_1, \bm U_m) \}   \bm 1\{U_{j_2}= N(V_2, \bm U_m) \}   \nonumber \\ 
& =  m^{\frac{2}{d}} \sum_{j_1=1}^m  ||V_1-U_{j_1}||   \bm 1\{U_{j_1}= N(V_1, \bm U_m) \} \sum_{j_2=1}^m ||V_2-U_{j_2}||  \bm 1\{U_{j_2}= N(V_2, \bm U_m) \}  \nonumber \\ 
& =  \mbb{E} \zeta_1(\pmb 0, m^{\frac{1}{d}}(\bm U_m-V_1))  \zeta_1(\pmb 0, m^{\frac{1}{d}}(\bm U_m-V_2)).  \nonumber 
\end{align}
Now, by arguments similar to the proof of \cite[Proposition 3.1]{yukichclt}, it follows that 
\begin{align*}
\lim_{M \rightarrow \infty}  m^{\frac{2}{d}} \mbb{E}(D_1D_2)=\lim_{M \rightarrow \infty}  \mbb{E} \zeta_1(\pmb 0, m^{\frac{1}{d}}(\bm U_m-V_1))  \zeta_1(\pmb 0, m^{\frac{1}{d}}(\bm U_m-V_2))  = \mbb{E} \zeta_1(\pmb 0, \mathcal P_{f_0(V)}))^2,
\end{align*} 
where, as before, $V$ is a random variable distributed according to the density $g$, and  $\mathcal P_{f_0(V)}$ is a Cox process with intensity measure $f_0(V)$. This combined with \eqref{eq:variance_D} and \eqref{eq:expDn}, shows that 
$$\Var\left(\frac{1}{n^{1-\frac{1}{d}}} \sum_{i=1}^n D_i \right)  \rightarrow 0.$$
This completes the proof of Lemma \ref{lm:D}. \hfill $\Box$ 

\subsection{Proof of Lemma \ref{lm:C}}
\label{sec:pf_C}

Denote $[m]:=\{1, 2, \ldots m\}$. To begin with note that $\mbb{E}(\bar C_{m, n})=\mbb{E}(C_1)$ and 
\begin{align}\label{eq:expC_I}
\mbb{E}(C_1)&=\sum_{j \in [m]} \sum_{s \in [m]\backslash\{j\}} \mbb{E} ||U_j-U_s||  \bm 1\{U_s= N(V_1, \bm U_m) \text{ and } U_j= N(U_s, \bm U_m) \} \nonumber \\
& = \zeta_2(V_1, \bm U_m). 
\end{align}
As in \eqref{eq:scaled_m}, by translation invariance, 
\begin{align}\label{eq:2zscaled_II}
\zeta_2(\pmb 0, m^{\frac{1}{d}}(\bm U_m-V_1)) 
&=m^{\frac{1}{d}} \sum_{j \in [m]} \sum_{s \in [m]\backslash\{j\}} ||U_j-U_s||  \bm 1\{U_s= N(V_1, \bm U_m) \text{ and } U_j= N(U_s, \bm U_m) \} \nonumber \\ 
&= m^{\frac{1}{d}} \zeta_2(V_1, \bm U_m). 
\end{align}
Now, as in Lemma \ref{lm:moment2},  it can be shown that $\sup_{m \in \mbb{N}}\mbb{E}\zeta_2(\pmb 0, m^{\frac{1}{d}}(\bm U_m-V_1))^2 \lesssim_d 1$. Therefore,  since the functional $\zeta_2(\cdot, \cdot)$ stabilizes on homogeneous Poisson processes, by arguments similar to the proof of \cite[Lemma 8.1]{yukichclt}, it follows that 
\begin{align}\label{eq:Mlimit_II}
\lim_{M \rightarrow \infty}\mbb{E} \zeta_2(\pmb 0, m^{\frac{1}{d}}(\bm U_m-V_1)) = \mbb{E} \zeta_2(\pmb 0, \mathcal P_{f_0(V)})) = f_0(y)^{-\frac{1}{d}}  \mbb{E} \zeta_2(\pmb 0, \mathcal P_1))
\end{align} 
where $\zeta_2(\pmb 0, \mathcal P_1)$ is as defined in equation (5.2), $V$ is a random variable distributed according to the density $g$, and $\mathcal P_{f_0(V)}$ is a Cox process with intensity measure $f_0(V)$. 
Then, recalling $n/m \rightarrow \rho$, and combining \eqref{eq:expC_I}, \eqref{eq:2zscaled_II}, and \eqref{eq:Mlimit_II} gives, 
$$\mbb{E}\left(\frac{1}{n^{1-\frac{1}{d}}} \sum_{i=1}^n C_i \right)   \rightarrow  \rho^{\frac{1}{d}}  \mbb{E} \zeta_2(\pmb 0, \mathcal P_1))  \int \frac{g(y)}{f_0(y)^{\frac{1}{d}} } \mathrm d y. $$
which establishes the limit in \eqref{eq:lm_C} in expectation. 

Finally, similar to the proof of Lemma \ref{lm:D}, it can be shown that the variance of the LHS in \eqref{eq:lm_C} goes to zero, completing  the proof. 
\hfill $\Box$ \\

As mentioned earlier, there does not appear to be  a closed form expression for $\zeta_2:=\mbb{E} \zeta_2(\pmb 0, \mathcal P_1))$. However, by an application of the FKG inequality  for Poisson processes \citep{janson_fkg,poisson_process_book}, it can be shown that $\zeta_2 \geq \zeta_1$. This is described in the following remark.

\begin{remark}\label{remark:12} From the definition of $\zeta_2$, we get 
	\begin{align}\label{eq:z2_calculation}
	\zeta_2:=\mbb{E} (\zeta_2(\pmb 0, \mathcal P_1)) & = \int  \int ||w'-b||  \P(b= N(0, \mc{P}_{1}^{0, b}) \text{ and } w' = N(b, \mc{P}_{1}^{0, b}\backslash\{0\}) ) \mathrm db \mathrm d w'. 
	\end{align}
	For $b, w' \in \mbb{R}^d$ fixed, consider the functions $\bm 1\{b= N(0, \mc{P}_{1}^{0, b})  \}$ and $\bm 1\{ w' = N(b, \mc{P}_{1}^{0, b}\backslash\{0\}) \}$, defined on the Poisson point process $\mc{P}_1^0$. Now, let $\Gamma$ and $\Gamma'$ be two realizations of the point process $\mc{P}_{1}^{0}$. Note that by if $\Gamma \subset \Gamma'$, then $\bm 1\{b= N(0, \Gamma')  \}  \leq \bm 1\{b= N(0, \Gamma)  \} $, because if $b$ is a nearest neighbor of the origin in $\Gamma'$, it will be also be nearest neighbor of the origin in the smaller set $\Gamma$. Similarly, for $\Gamma \subset \Gamma'$, $\bm 1\{ w' = N(b, \Gamma'\backslash\{0\}) \} \leq \bm 1\{ w' = N(b, \Gamma\backslash\{0\}) \}$. Therefore, both the functions $\bm 1\{b= N(0, \mc{P}_{1}^{0, b})  \}$ and $\bm 1\{ w' = N(b, \mc{P}_{1}^{0, b}\backslash\{0\}) \}$ are nonincreasing, and by an application of the FKG inequality for functions on Poisson processes \citep[Lemma 2.1]{janson_fkg}, it follows that 
	$$\P(b= N(0, \mc{P}_{1}^{0, b}) \text{ and } w' = N(b, \mc{P}_{1}^{0, b}\backslash\{0\}) )  \geq \P(b= N(0, \mc{P}_{1}^{0})) \P(w'= N(b, \mc{P}_{1}^{b}))$$
	This combined with \eqref{eq:z2_calculation} gives,  
	\begin{align*}
	\zeta_2& \geq \int \int ||w'-b|| \bm \P(b= N(0, \mc{P}_{1}^{0})) \P(w'= N(b, \mc{P}_{1}^{b})) \mathrm d b \mathrm d w' \nonumber \\
	&=  \int \int ||w'-b||  e^{-V_d ||b||}  e^{-V_d ||w'-b||}  \mathrm d b \mathrm d w'  \nonumber \\ 
	&= \left(\int  e^{-V_d ||b||} \mathrm d b \right) \left( \int ||v||   e^{-V_d ||v||}   \mathrm d v \right) = \zeta_1,
	\end{align*}
	where the last step uses the definition of $\zeta_1$ from equation (5.4), and $\int  e^{-V_d ||b||} \mathrm d b= S_d \int_0^\infty r^{d-1} e^{-V_d r^d} \mathrm d r = V_d \int_0^\infty  e^{-V_d y} \mathrm d y=1$. 
\end{remark}

\begin{table}[h]
	\centering
	\small 
	\begin{tabular}{|c||c|c|c|}
		\hline 
		$d$ & $\zeta_1$ & $\zeta_2$ & $\Delta_d=\zeta_2-\zeta_1$ \\
		\hline
		1 & 0.5006 & 0.7493 & 0.2487 \\
		\hline
		2 & 0.5008 & 0.5969  & 0.0961 \\
		\hline
		3 & 0.5580 & 0.6155 & 0.0574 \\
		\hline
		4 & 0.6187 & 0.6572 & 0.0385 \\
		\hline
		5 & 0.6782 & 0.7054 & 0.0271 \\ 
		\hline
		6 & 0.7361 & 0.7548 & 0.0187 \\
		\hline
	\end{tabular}
	\caption{\small{Numerical estimates of $\zeta_1$ and $\zeta_2$.}}
	\label{table:values_12}
\end{table}\vspace{-0.1in}
\normalsize

Numerical estimates of the constants $\zeta_1$ and $\zeta_2$ for small dimensions are given in Table \ref{table:values_12}. This is computed using the average (over 20 iterations) of the values of $n^{\frac{1}{d}} \bar D_{m, n}$ and  $n^{\frac{1}{d}} \bar C_{m, n}$ (recall equation (2.7)) with $m=n=100000$ i.i.d uniform points in the $d$-dimensional unit cube $[0, 1]^d$. 

\section{Proof of Corollary 1}
\label{sec:pfcorH0H1}

Note that, by equation (5.7), for  $g \in \mc{F}(\bm \theta)$, $T_{m, n}  \pto \varphi(f_0, g, \rho) < \gamma$. This implies, $\lim_{m, n \rightarrow \infty}\P_{f_0, g}(T_{m, n} > \gamma) = 0$, which proves (5.8).

Under the alternative, suppose $g(y) = \sum_{a=1}^K \bar \lambda_a p(y|\theta_a')$, such that, for some $1 \leq j \leq K$ with $\bar \lambda_j >0$,  $\min_{1 \leq a \leq K}||\theta_j'-\theta_a|| \geq \varepsilon(\gamma)$, where $\varepsilon(\gamma)$ will be chosen later. 
Then 
\begin{align}\label{eq:alt_bound}
\varphi(f_0, g, \rho)= \rho^{\frac{1}{d}} \Delta_d \int \frac{g(y)}{f_0(y)^{\frac{1}{d}}} \mathrm d y  & =  \rho^{\frac{1}{d}} \Delta_d \sum_{a=1}^K \int \frac{\bar \lambda_a p(y|\theta_a') }{  \left(\sum_{b=1}^K  w_b p(y|\theta_b)\right)^{\frac{1}{d}}} \mathrm d y  \nonumber \\ 
& \geq  \rho^{\frac{1}{d}} \Delta_d \int_{B(\theta_j', 1)} \frac{\bar \lambda_j p(y|\theta_j') }{ \left(\sum_{b=1}^K   w_b p(y|\theta_b) \right)^{\frac{1}{d}}} \mathrm d y. 
\end{align}  

Now, since the function $r(\cdot)$ is uniformly continuous and $\int_0^\infty r(z) \mathrm d z < \infty$, it follows that $\lim_{z \rightarrow \infty} r(z)=0$ (see discussion following \cite[Corollary 1]{density_limit}). This implies for every $M>0$ there exists a $\eta(M, d)>0$, such that $r(z) \leq M^{-\frac{1}{d}}$, for $z > \eta(M, d)$. Define $$M:=\frac{2 \gamma}{\rho^{\frac{1}{d}} \Delta_d  L  \int_{B(0, 1)} p(y) \mathrm dy} \quad \text{and} \quad \varepsilon(\gamma):=\eta(M, d)+1.$$ Take a point $\theta_j'$ such that $||\theta_j'-\theta_a|| \geq \varepsilon(\gamma)$, for all $1 \leq a \leq K$. Then, for all $1 \leq a \leq K$, if $y \in B(\theta_j', 1)$, 
$$\eta(M, d)+1 \leq ||\theta_j'-\theta_a|| \leq ||\theta_j'-y|| + ||y-\theta_a|| \leq 1 + ||y-\theta_a||,$$
implies $||y-\theta_a|| \geq \eta(M, d)$. Therefore, for all $1 \leq a \leq K$, if $y \in B(\theta_j', 1)$,  $p(y|\theta_a) = p(y-\theta_a)= r(||y-\theta_a||) \leq M^{-\frac{1}{d}} $ and $\sum_{a=1}^K   w_ap(y|\theta_a) \leq M^{-\frac{1}{d}}$. Then, from \eqref{eq:alt_bound}, 
\begin{align*}
\varphi(f_0, g, \rho)  \geq \rho^{\frac{1}{d}}\Delta_d \bar \lambda_j  M \int_{B(\theta_j', 1)} p(y|\theta_j') \mathrm d y & = \rho^{\frac{1}{d}} \Delta_d \bar \lambda_j  M \int_{B(\theta_j', 1)} p(y-\theta_j') \mathrm d y \nonumber \\ 
& \geq \rho^{\frac{1}{d}} \Delta_d  L M  \int_{B(0, 1)} p(y) \mathrm dy  \nonumber \\ 
& = 2 \gamma. 
\end{align*}
This implies $\lim_{m, n \rightarrow \infty}\P_{f_0, g}(T_{m, n} > \gamma) =1$, since $T_{m, n}  \pto \varphi(f_0, g, \rho) > 2\gamma$, for $g$ as above. This completes the proof of (5.9). \hfill $\Box$ \\ 

Note that the separation $\varepsilon(\gamma)$ depends on $\eta(M, d)$, the rate of decay of the tail of the base density $p$. For instance, when $p$ is the standard multivariate normal distribution $N(0, \mathrm I_d)$, then it suffices to choose $\eta(M, d)=K(d) \sqrt{\log M}$, where $K(d)$ is a constant depending on $d$.

\section{Additional Numerical Experiments}
\label{sec:additionalsims}
\subsection{Sensitivity of the Numerical Experiments in Section 3 to the choice of $\tau_{fc}$}
\label{sec:sim_tau}
\red{We consider the setting of Experiment 2 in section 3.2 and report the sensitivity of our inference using \texttt{TRUH} to changes in the fold change constant $\tau_{fc}$. Recall that in Experiment 2, $$F_0=0.5~{\mathrm{Gam}}_d(\text{shape}=5\bm 1_d,\text{rate}=\bm 1_d,\bm \Sigma_1)+0.5~{\mathrm{Exp}}_d(\text{rate}=\bm 1_d,\bm \Sigma_2),$$ 
	where ${\mathrm{Gam}}_d$ and ${\mathrm{Exp}}_d$ are $d$ dimensional Gamma and Exponential distributions. For generating correlated Gamma and Exponential variables, we use the Gaussian copula approach based function from the R-package \texttt{lcmix} \citep{lcmix,xue2000multivariate}. 
	We consider tapering matrices with positive and negative autocorrelations: $(\bm \Sigma_1)_{ij}=0.7^{|i-j|}$ and $(\bm \Sigma_2)_{ij}=-0.9^{|i-j|}$ for $1\le i,j\le d$.
	For simulating $\bm V_n$ from $G$, we consider the following two scenarios:}
\begin{itemize}
	
	\item \red{Scenario I: Here, $G={\mathrm{Exp}}_d(\text{rate}=\bm 1_d,\bm \Sigma_2)$. In this case, $G$ arises from only one of the components of $F_0$, that is, $G\in\mc{F}(F_0)$.}
	
	\item \red{Scenario II: Here, $G=0.1~{\mathrm{Gam}}_d(\text{shape}=10\bm 1_d,\text{rate}=0.5\bm 1_d,\bm \Sigma_1)+0.9~{\mathrm{Exp}}_d(\text{rate}=\bm 1_d,\bm \Sigma_2)$. In this setting, $G\notin\mc{F}(F_0)$ and the composite null $H_0$ is not true. When the ratio $n/m$ is small, this scenario presents a difficult setting for detecting departures from $H_0$ as majority of the case samples from $\bm V_n$ will arise from ${\mathrm{Exp}}_d(\text{rate}=\bm 1_d,\bm \Sigma_2)$ and the tests will rely on only a small fraction of samples from ${\mathrm{Gam}}_d(\text{shape}=10\bm 1_d,\text{rate}=0.5\bm 1_d,\bm \Sigma_1)$ to reject the null hypothesis.} 
\end{itemize}
\begin{table}[htbp]
	\centering
	\caption{Rejection rates of \texttt{TRUH} at $5\%$ level of significance: Experiment 2 and Scenario I wherein $H_0: G \in \mathcal{F}(F_0)$ is true.}
	\begin{tabular}{lcccccc}
		\hline
		& \multicolumn{3}{c}{$m = 500,~n = 50$} & \multicolumn{3}{c}{$m = 2000,~n = 200$} \\
		\hline
		\multicolumn{1}{c}{$\tau_{fc}$} & $d=5$   & $d=15$  & $d=30$  & $d=5$   & $d=15$  & $d=30$ \\
		\hline
		$1.0$ & 0.000 & 0.000 & 0.000 & 0.000 & 0.000 & 0.000 \\
		$1.2$ & 0.000 & 0.000 & 0.000  & 0.000 & 0.000 & 0.000 \\
		$1.4$ & 0.000 & 0.000 & 0.000  & 0.000 & 0.000 & 0.000 \\
		\hline
	\end{tabular}%
	\label{tab:exp2set1_1}%
\end{table}%
\begin{table}[htbp]
	\centering
	\caption{Rejection rates of $\texttt{TRUH}$ at $5\%$ level of significance: Experiment 2 and Scenario II wherein $H_0: G \in \mathcal{F}(F_0)$ is false.}
	\begin{tabular}{lcccccc}
		\hline
		& \multicolumn{3}{c}{$m = 500, n = 10$} & \multicolumn{3}{c}{$m = 2000, n = 40$} \\
		\hline
		\multicolumn{1}{c}{$\tau_{fc}$} & $d=5$   & $d=15$  & $d=30$  & $d=5$   & $d=15$  & $d=30$\\
		\hline
		$1.0$ & 0.580 & 0.580 & 0.580 & 0.880 & 0.940 & 0.960 \\
		$1.2$ & 0.500 & 0.560 & 0.580  & 0.820 & 0.860 & 0.900 \\
		$1.4$ & 0.460 & 0.480 & 0.500  & 0.780 & 0.760 & 0.700 \\
		\hline
	\end{tabular}%
	\label{tab:exp2set2_1}%
\end{table}%
\red{Tables \ref{tab:exp2set1_1} and \ref{tab:exp2set2_1} report the average rejection rates of \texttt{TRUH} across $100$ repetitions of the test as $\tau_{fc}$ varies over $\{1,1.2,1.4\}$. 
	We note that the rejection rates under Scenario II are bigger than those of Scenario I,  which indicates that our proposed procedure is powerful against departures from the null hypothesis while the rejection rates under Scenario I are below the prespecified 
	$0.05$ level establishing that it is a conservative test across all the regimes considered in the table. These results also indicate that an appropriate choice of $\tau_{fc}$ must be bigger or equal to $1$ for a value less than $1$ may lead to incorrect rejections of the null hypothesis.}
\begin{table}[!t]
	\centering
	\caption{Rejection rates of \texttt{TRUH} at $5\%$ level of significance under Dirichlet sampling of mixing proportions: \newline Experiment 2 and Scenario I wherein $H_0: G \in \mathcal{F}(F_0)$ is true.}
	\begin{tabular}{lcccccc}
		\hline
		& \multicolumn{3}{c}{$m = 500,~n = 50$} & \multicolumn{3}{c}{$m = 2000,~n = 200$} \\
		\hline
		\multicolumn{1}{c}{$\tau_{fc}$} & $d=5$   & $d=15$  & $d=30$  & $d=5$   & $d=15$  & $d=30$ \\
		\hline
		$1.0$ & 0.000 & 0.000 & 0.000 & 0.000 & 0.000 & 0.000 \\
		$1.2$ & 0.000 & 0.000 & 0.000  & 0.000 & 0.000 & 0.000 \\
		$1.4$ & 0.000 & 0.000 & 0.000  & 0.000 & 0.000 & 0.000 \\
		\hline
	\end{tabular}%
	\label{tab:exp2set1_2}%
\end{table}%
\begin{table}[htbp]
	\centering
	\caption{Rejection rates of $\texttt{TRUH}$ at $5\%$ level of significance under Dirichlet sampling of mixing proportions:\newline Experiment 2 and Scenario II wherein $H_0: G \in \mathcal{F}(F_0)$ is false.}
	\begin{tabular}{lcccccc}
		\hline
		& \multicolumn{3}{c}{$m = 500, n = 10$} & \multicolumn{3}{c}{$m = 2000, n = 40$} \\
		\hline
		\multicolumn{1}{c}{$\tau_{fc}$} & $d=5$   & $d=15$  & $d=30$  & $d=5$   & $d=15$  & $d=30$\\
		\hline
		$1.0$ & 0.580 & 0.580 & 0.580 & 0.880 & 0.940 & 0.960 \\
		$1.2$ & 0.500 & 0.560 & 0.580  & 0.840 & 0.860 & 0.900 \\
		$1.4$ & 0.440 & 0.480 & 0.500  & 0.780 & 0.760 & 0.700 \\
		\hline
	\end{tabular}%
	\label{tab:exp2set2_2}%
\end{table}%

\red{The preceding analyses were based on sampling the mixing proportions $\{\lambda_1,\ldots,\lambda_{\hat{K}}\}$ only from the corners of the $\hat{K}$ dimensional simplex $\mc{S}_{\hat{K}}$ (see section 2.3). In what follows, we alter this sampling scheme and sample the $\hat{K}$ mixing proportions from a Dirichlet distribution with parameters $\{\beta_1,\ldots,\beta_{\hat{K}}\}$. We set $\beta_a=0.1$ for $1\le a\le \hat{K}$ and report the rejection rates under this sampling scheme in tables \ref{tab:exp2set1_2} and \ref{tab:exp2set2_2}. With $\beta_a=0.1$, the Dirichlet distribution places a large mass on the corners of the $\hat{K}$ dimensional simplex which explains the similar trend in the rejections rates that is observed across both the scenarios in tables \ref{tab:exp2set1_2} and \ref{tab:exp2set2_2} when compared to tables \ref{tab:exp2set1_1} and \ref{tab:exp2set2_1}, respectively.}  
\subsection{Sensitivity of the Real Data Analysis in Section 4 to the choice of $\tau_{fc}$}
\label{sec:realdata_tau}
\red{In section 4 the fold change constant $\tau_{fc}$ was set to $1.1$. In this section we report the sensitivity of the results reported in section 4 for $\tau_{fc}\in\{1,1.1,1.3,1.7,2\}$. In tables \ref{tab:case1_taufc} to \ref{tab:case4_taufc} we report the p-values of the \texttt{TRUH} test statistic for testing remodeling under HIV infection under the four cases as described in section 4. For Cases 1 and 3, which are known to exhibit remodeling, we note from tables \ref{tab:case1_taufc} and \ref{tab:case3_taufc} that for $\tau_{fc}>1.1$, the TRUH test statistic fails to detect remodeling across all four donors. This is not unexpected since a relatively large value of $\tau_{fc}$ offers higher conservatism in rejecting the null hypothesis of no remodeling. Tables \ref{tab:case2_taufc} and \ref{tab:case4_taufc}, on the other hand, represent cases of no remodeling and $\tau_{fc}\ge 1.1$ allows \texttt{TRUH} to correctly detect no remodeling for Cases 2 and 4. }
\begin{table}[!t]
	\centering
	\caption{p-values of \texttt{TRUH} in CASE 1: Uninfected versus \texttt{Nef-rich} HIV Infected for entire 35 markers.}
	\scalebox{0.8}{
		\begin{tabular}{lcccc}
			\hline
			& \multicolumn{1}{c}{Donor 1} & \multicolumn{1}{c}{Donor 2} & \multicolumn{1}{c}{Donor 3} & \multicolumn{1}{c}{Donor 4} \\
			$\tau_{fc}$ & \multicolumn{1}{l}{$m = 24,984, n = 245$} & \multicolumn{1}{l}{$m = 31,552, n = 521$} & \multicolumn{1}{l}{$m = 17,704, n = 211$} & \multicolumn{1}{l}{$m = 22,830, n = 660$} \\
			\hline
			$1.0$ &  $<0.001$    & $<0.001$      &  $<0.001$     & $<0.001$ \\
			$1.1$ & $<0.001$      &   $<0.001$    &  $<0.001$     & $<0.001$ \\
			$1.3$ &   $0.082$    & $0.05$      &  $0.086$     & $0.292$ \\
			$1.7$ &  $1$    & $1$      &  $1$     & $1$ \\
			$2.0$ &   $1$    & $1$      &  $1$     & $1$  \\
			\hline
	\end{tabular}}%
	\label{tab:case1_taufc}%
	\vspace{10pt}
	\caption{p-values of \texttt{TRUH} in CASE 2: Uninfected versus \texttt{Nef-rich} HIV Infected for 31 invariant markers.}
	\scalebox{0.8}{
		\begin{tabular}{lcccc}
			\hline
			& \multicolumn{1}{c}{Donor 1} & \multicolumn{1}{c}{Donor 2} & \multicolumn{1}{c}{Donor 3} & \multicolumn{1}{c}{Donor 4} \\
			$\tau_{fc}$ & \multicolumn{1}{l}{$m = 24,984, n = 245$} & \multicolumn{1}{l}{$m = 31,552, n = 521$} & \multicolumn{1}{l}{$m = 17,704, n = 211$} & \multicolumn{1}{l}{$m = 22,830, n = 660$} \\
			\hline
			$1.0$ &  $<0.001$    & $<0.001$      &  $<0.001$     & $<0.001$ \\
			$1.1$ & $<0.001$      &   $0.67$    &  $0.274$     & $0.914$ \\
			$1.3$ &   $1$    & $1$      &  $1$     & $1$ \\
			$1.7$ &  $1$    & $1$      &  $1$     & $1$ \\
			$2.0$ &   $1$    & $1$      &  $1$     & $1$  \\
			\hline
	\end{tabular}}
	\label{tab:case2_taufc}%
	\vspace{10pt}
	\caption{p-values of \texttt{TRUH} in CASE 3: Uninfected versus \texttt{Nef-deficient} HIV Infected for the entire 35 markers.}
	\scalebox{0.8}{
		\begin{tabular}{lcccc}
			\hline
			& \multicolumn{1}{c}{Donor 1} & \multicolumn{1}{c}{Donor 2} & \multicolumn{1}{c}{Donor 3} & \multicolumn{1}{c}{Donor 4} \\
			$\tau_{fc}$ & \multicolumn{1}{l}{$m = 24,984, n = 129$} & \multicolumn{1}{l}{$m = 31,552, n = 382$} & \multicolumn{1}{l}{$m = 17,704, n = 174$} & \multicolumn{1}{l}{$m = 22,830, n = 440$} \\
			\hline
			$1.0$ &  $<0.001$    & $<0.001$      &  $<0.001$     & $<0.001$ \\
			$1.1$ & $<0.001$      &   $<0.001$    &  $<0.001$     & $<0.001$ \\
			$1.3$ &   $1$    &   $1$    &  $1$     & $1$ \\
			$1.7$ &  $1$    &   $1$    &  $1$     & $1$ \\
			$2.0$ &   $1$    &  $1$     &  $1$     & $1$  \\
			\hline
	\end{tabular}}
	\label{tab:case3_taufc}%
	\vspace{10pt}
	\caption{p-values of \texttt{TRUH} in CASE 4: Uninfected versus \texttt{Nef-deficient} HIV Infected for the entire 31 invariant markers.}
	\scalebox{0.8}{
		\begin{tabular}{lcccc}
			\hline
			& \multicolumn{1}{c}{Donor 1} & \multicolumn{1}{c}{Donor 2} & \multicolumn{1}{c}{Donor 3} & \multicolumn{1}{c}{Donor 4} \\
			$\tau_{fc}$ & \multicolumn{1}{l}{$m = 24,984, n = 129$} & \multicolumn{1}{l}{$m = 31,552, n = 382$} & \multicolumn{1}{l}{$m = 17,704, n = 174$} & \multicolumn{1}{l}{$m = 22,830, n = 440$} \\
			\hline
			$1.0$ &  $<0.001$    & $<0.001$      &  $<0.001$     & $<0.001$ \\
			$1.1$ & $0.58$      &   $0.94$    &  $0.464$     & $0.524$ \\
			$1.3$ &   $1$    & $1$      &  $1$     & $1$ \\
			$1.7$ &  $1$    & $1$      &  $1$     & $1$ \\
			$2.0$ &   $1$    & $1$      &  $1$     & $1$  \\
			\hline
	\end{tabular}}
	\label{tab:case4_taufc}%
\end{table}%
\subsection{Computation Time Comparisons}
\label{sec:timing} 
\red{In this section we present a comparison of the computing time for each of the seven competing testing procedures under the settings of Scenarios I and II of Experiment 2 (see section 3.2) with $n=15000,~m=150$ and $d\in\{5,15,30,50\}$. Tables \ref{tab:timing} and \ref{tab:timing_2} report the average computing time in minutes across 10 repetitions of the testing problem as $d$ varies. Here the computing time represents the time each test takes to generate a p-value. We note that the \texttt{Energy} test is extremely efficient with an average computation time just over a minute for these scenarios. Our proposed testing procedure \texttt{TRUH} is the next best and is closely followed by the four variants of the \texttt{Edge Count} tests. The \texttt{R} package \texttt{gtests} that implements these variants of the \texttt{Edge Count} tests, spits the results for all these variants simultaneously and thus the different \texttt{Edge Count} tests exhibit the same performance in tables \ref{tab:timing} and \ref{tab:timing_2}. We note that the computation time of \texttt{TRUH} increases with $d$, which is not surprising because the computational cost for running the 1-nearest neighbor algorithm is $O(nmd)$. The \texttt{Edge Count} tests, on the other hand, rely on a minimum spanning tree construction which has $O(n^2)$ complexity and dominates the overall running time. In our experiments, we find that the \texttt{Crossmatch} test is the slowest primarily because this test requires a number of computationally expensive steps such as ranking each of the $d$ dimensions, computing and inverting the $d$ dimensional covariance matrix of the ranks and finally calculating the $(n+m)\times (n+m)$ matrix of Mahalanobis distance between the rank pairs.}
\begin{table}[!t]
	\centering
	\caption{\small Mean computing time (in minutes) for the seven competing testing procedures under the setting of Scenario I Experiment 2. Here $n=15000,~m=150$.}
	\scalebox{0.9}{
		\begin{tabular}{lccccccc}
			\hline
			\multicolumn{1}{c}{$d$} & \texttt{Energy}  & \texttt{Crossmatch}  & \texttt{E Count}  & \texttt{GE Count}   & \texttt{WE Count}  & \texttt{MTE Count} & \texttt{TRUH}\\
			\hline
			$5$& 1.128 & 41.221 & 8.242 & 8.242 & 8.242 & 8.242 &1.047 \\
			$15$ & 1.142 & 35.225 & 9.394 & 9.394 & 9.394 & 9.394 & 2.161 \\
			$30$ & 1.162 & 35.141 &9.481 & 9.481 & 9.481 & 9.481 &3.970 \\
			$50$ & 1.235 & 42.062 & 9.784  & 9.784 & 9.784 & 9.784 & 6.380\\
			\hline
	\end{tabular}}%
	\label{tab:timing}%
\end{table}
\begin{table}[!t]
	\centering
	\caption{\small Mean computing time (in minutes) for the seven competing testing procedures under the setting of Scenario II Experiment 2. Here $n=15000,~m=150$.}
	\scalebox{0.9}{
		\begin{tabular}{lccccccc}
			\hline
			\multicolumn{1}{c}{$d$} & \texttt{Energy}  & \texttt{Crossmatch}  & \texttt{E Count}  & \texttt{GE Count}   & \texttt{WE Count}  & \texttt{MTE Count} & \texttt{TRUH}\\
			\hline
			$5$& 1.107 & 45.604 & 7.277 & 7.277 & 7.277 & 7.277 &1.003 \\
			$15$ & 1.147 & 37.688 & 7.502 & 7.502 & 7.502 & 7.502 & 2.099 \\
			$30$ & 1.117 & 32.858 & 7.581 & 7.581 & 7.581 & 7.581 & 3.634 \\
			$50$ & 1.150 & 46.575 & 6.833  & 6.833 & 6.833 & 6.833 & 6.044\\
			\hline
	\end{tabular}}%
	\label{tab:timing_2}%
\end{table}
\subsection{Influence of mixing proportions on the null distribution}
\label{sec:influence} 
\red{To understand how the specific mixing proportions influence the null distribution of the \texttt{TRUH} statistic, we consider the setting of Experiment 2 (section 3.2 of the manuscript) where the uninfected cells arise from a 2-component mixture distribution given by $F_0=0.5~{\mathrm{Gam}}_d(\text{shape}=5\bm 1_d,\text{rate}=\bm 1_d,\bm \Sigma_1)+0.5~{\mathrm{Exp}}_d(\text{rate}=\bm 1_d,\bm \Sigma_2)$. The infected cells arise from $G={\mathrm{Exp}}_d(\text{rate}=\bm 1_d,\bm \Sigma_2)$ for scenario I and from $G=0.1\mathrm{Gam}_d(\text{shape}=10\bm 1_d,\text{rate}=0.5\bm 1_d,\bm \Sigma_1)+0.9\mathrm{Exp}_d(\text{rate}=\bm 1_d,\bm \Sigma_2)$ for scenario II. Thus scenario I represents preferential infection and scenario II represents remodeling. Here ${\mathrm{Gam}}_d$ and ${\mathrm{Exp}}_d$ are $d$ dimensional Gamma and Exponential distributions.} 

\red{Figure \ref{fig:exp3_1} plots the $\hat{K}=2$ null distributions of the \texttt{TRUH} statistic under Scenario I. Here the bootstrap algorithm generates the mixing proportions $\{\lambda_1, \ldots, \lambda_{\hat K}\}$ in step (i) only from the corners of the $\hat{K}=2$ dimensional simplex $\mc{S}_{\hat{K}}$ and so there are $B_1=\hat{K}$ such null distributions for each of the $50$ repetitions of the test. The null distributions represented by the red box plots arise when the $n$ pseudo infected cells are randomly sampled from the first component of $F_0$, which corresponds to a configuration of mixing proportions given by $\{\lambda_1=1, \lambda_2=0\}$. The box plots in green, on the other hand, arise when the $n$ pseudo infected cells are randomly sampled from the second component of $F_0$, which corresponds to $\{\lambda_1=0, \lambda_2=1\}$. The blue dots represent the \texttt{TRUH} test statistic across the $50$ repetitions of the test. 
	Note that in Scenario I, the infected cells arise from the second component of $F_0$ and this explains why the \texttt{TRUH} test statistic, given by the blue dots, are closer to the box plots in green. However, the null distribution represented by the red box plots, which correspond to the mixing proportions $\{\lambda_1=1, \lambda_2=0\}$, offer more conservatism in rejecting the null hypothesis of no remodeling as far as Scenario I is concerned. }

\begin{figure}[!t]
	\centering
	\includegraphics[width=1\linewidth]{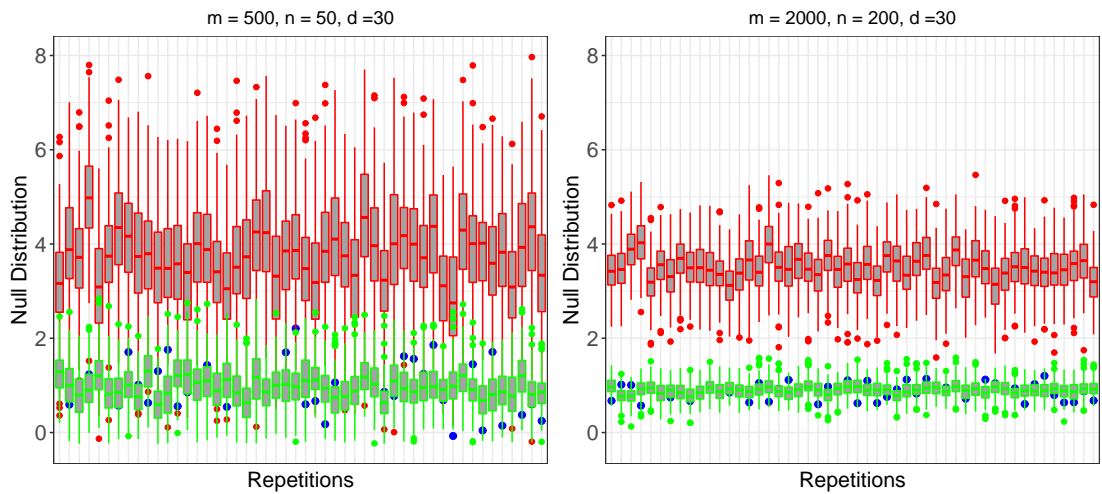}
	\caption{\small The red box plots  and the green box plots represent $\hat{K}=2$ null distributions for each of the $50$ repetitions of Scenario I ($H_0$ is true) corresponding to mixing proportions $\{\lambda_1=1,\lambda_2=0\}$ and $\{\lambda_1=0,\lambda_2=1\}$ respectively. The blue dots are the \texttt{TRUH} test statistic in each repetition of the test. Left: Uninfected and Infected sample sizes are $m=500, n=50$ and dimensionality of each sample is $d=30$. Right: $m=2000, n=200, d=30$.} 
	\label{fig:exp3_1}
\end{figure}
\begin{figure}[!h]
	\centering
	\includegraphics[width=1\linewidth]{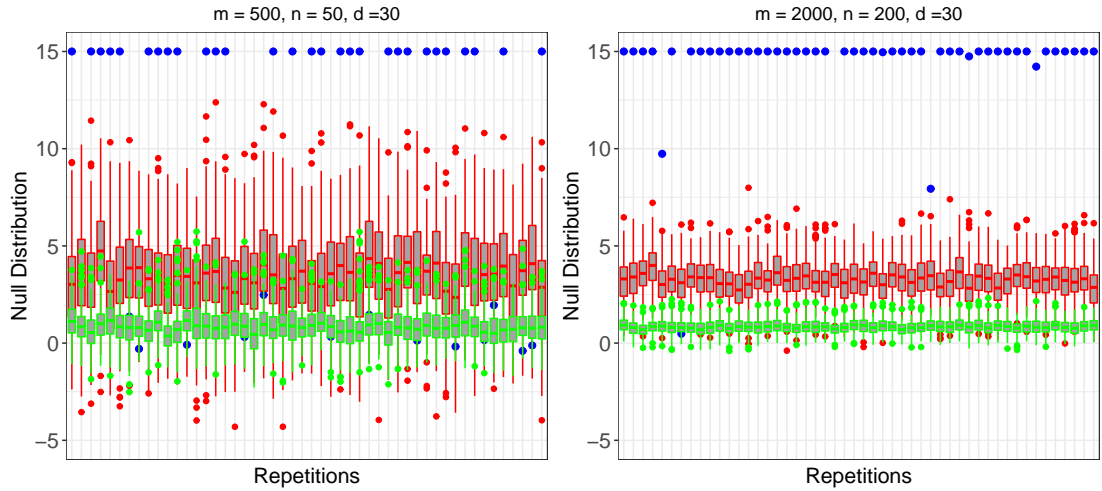}
	\caption{\small The red box plots  and the green box plots represent $\hat{K}=2$ null distributions for each of the $50$ repetitions of Scenario II ($H_0$ is false) corresponding to mixing proportions $\{\lambda_1=1,\lambda_2=0\}$ and $\{\lambda_1=0,\lambda_2=1\}$ respectively. The blue dots are the \texttt{TRUH} test statistic in each repetition of the test and are capped at $15$ for ease of visual representation. Left: Uninfected and Infected sample sizes are $m=500, n=50$ and dimensionality of each sample is $d=30$. Right: $m=2000, n=200, d=30$.}
	\label{fig:exp3_2}
\end{figure}
\red{Figure \ref{fig:exp3_2}, on the other hand, plots the $\hat{K}=2$ null distributions of the \texttt{TRUH} statistic under Scenario II which is a case of remodeling. Here the \texttt{TRUH} test statistic is capped at $15$ for ease of visual representation. Unlike Scenario I, the infected cells do not arise from the convex hull $\mc{F}(F_0)$ of $F_0$, which explains why the \texttt{TRUH} test statistic are far away from the two null distributions. However, even in this scenario when the $n$ pseudo infected cells are randomly sampled from the first component of $F_0$, the corresponding null distribution (box plots in red) offers more conservatism in rejecting the null hypothesis of no remodeling than the null distribution that arises under the configuration $\{\lambda_1=0, \lambda_2=1\}$.}

\bibliographystyle{chicago}
\bibliography{ref,single-cell,paper-ref}

%
%
%

\end{document}